\documentclass{amsart}
\usepackage{epsfig}
\usepackage{graphics}
\usepackage{amsfonts}
\usepackage{amsmath}
\usepackage{amssymb,mathrsfs}
\usepackage{latexsym}

\newtheorem{theorem}{Theorem}[section]
\newtheorem{lemma}{Lemma}[section]
\newtheorem{proposition}{Proposition}[section]
\newtheorem{corollary}{Corollary}[section]

\theoremstyle{definition}

\newtheorem{remark}{Remark}[section]

\DeclareMathOperator{\supp}{Supp}
\DeclareMathOperator{\Inte}{Int}
\RequirePackage[colorlinks,citecolor=blue,urlcolor=blue]{hyperref}
\allowdisplaybreaks[4]
\numberwithin{equation}{section}

\begin{document}
\title{Multifractal analysis of complex random cascades}
\author{Julien Barral}
\address{INRIA Rocquencourt, B.P. 105, 78153 Le Chesnay Cedex, France}
\email{Julien.Barral@inria.fr}

\author{Xiong Jin}
\address{INRIA Rocquencourt, B.P. 105, 78153 Le Chesnay Cedex, France}
\email{Xiong.Jin@inria.fr}

\begin{abstract}
We achieve the multifractal analysis of a class of complex valued statistically self-similar continuous functions.  For we use multifractal formalisms associated with pointwise oscillation exponents of all orders. Our  study exhibits new phenomena in multifractal analysis of continuous functions. In particular, we find examples of statistically self-similar such functions obeying the multifractal formalism and for which the support of the singularity spectrum is the whole interval $[0,\infty]$.
\end{abstract}

\maketitle

\section{Introduction}\label{intro}

This paper deals with the multifractal formalism for functions and the multifractal analysis of a new class of statistically self-similar functions introduced in~\cite{BJMpartII}. This class is the natural extension to  continuous functions of the random measures introduced in \cite{M2} and considered as a fundamental example of multifractal signals model since the notion of multifractality has been explicitely formulated  \cite{hentschel,FrischParisi,HaJeKaPrSh} (see also \cite{KP,HoWa,CK,Mol,B2} for the multifractal analysis and thermodynamical interpretation of these measures).  While the measures contructed in  \cite{M2} provide a model for the energy dissipation in a turbulent fluid, the functions we consider may be used to model the temporal fluctuations of the speed measured at a given point of the fluid. Also, they provide an alternative to models of multifractal signals which use multifractal measures, either to make a multifractal time change in Fractional Brownian motions \cite{Mandfin,BacryMuzy}, or to build wavelet series \cite{ABM,BSw}.

We exhibit statistically self-similar continuous functions possessing the remarkable property to obey the multifractal formalism, and simultaneously to be nowhere locally H\"older continuous. Specifically, the support of their multifractal spectra does contain the exponent 0, and the set  of points at which the pointwise H\"older exponent is 0 is dense in the support of the function. Moreover, these spectra can also be left-sided with singularity spectra supported by the whole interval $[0,\infty]$ (see Figure~\ref{left}). These properties are new phenomena in multifractal analysis of continuous self-similar functions. Let us explain this in detail, by starting with some recalls and remarks on multifractal analysis of functions.

\medskip

Multifractal analysis is a natural framework to describe geometrically the heterogeneity in the distribution at small scales of the H\"older singularities of a given locally bounded function or signal $f:U\subset \mathbb{R}^n\to \mathbb{R}^p$ ($n,p\ge 1$). In this paper, we will work in dimension 1 with continuous functions $f: I\to \mathbb{R}$ (or $\mathbb{C}$), where $I$ is a compact interval. The most natural notion of H\"older singularity is the pointwise H\"older exponent, which finely describes the best polynomial approximation of $f$ at any point $t_0\in I$ and is defined by 
$$
h_f(t_0)=\sup\{h\ge 0:\exists P\in \mathbb{C}[t], |f(t)-f(t_0)-P(t-t_0)|=O(|t-t_0|^h), t\to t_0\}.
$$
Then, the multifractal analysis of $f$ consists in computing the Hausdorff dimension of the H\"older singularities level sets, also called iso-H\"older sets 
$$
E_f(h)=\{t\in I: h_f(t)=h\},\quad  h\ge 0.
$$ 
The mapping $h\ge 0\mapsto \dim_H E_f(h)$ is called the {\it singularity spectrum} of $f$ ($\dim_{H}$ stands for the Hausdorff dimension, whose definition is recalled at the end of this section); the support of this spectrum is the set of those $h$ such that $E_f(h)\neq\emptyset$. The function is called multifractal when at least two iso-H\"older sets are non-empty. Otherwise, it is called monofractal. 

When the function $f$ is globally H\"older continuous, it has been proved in \cite{JAFFNOTE,Jafw} that the exponent $h_f(t)$ can always be obtained through the asymptotic behavior of the wavelet coefficients of $f$ located in a neighborhood of $t$, when the wavelet is smooth enough. Then, wavelet expansions  have been used successfully to characterize the iso-H\"older sets of wide classes of functions \cite{JAFFrev,Jafsiam1,Ben,JAFFAAP,JAFFJMPA,Aubry,BSw,Durand}, sometimes directly constructed as wavelet series (expansions in Schauder's basis have also been used \cite{JafMand1}).

For most of these functions, the singularity spectrum can be obtained as the Legendre transform of a free energy function computed on the wavelet coefficients. This is the so-called multifractal formalism for functions, studied and developed rigorously in \cite{Jafsiam1,JAFFJMP,JAFFJMPA,Jafw} after being introduced by physicists \cite{hentschel,FrischParisi,HaJeKaPrSh,muzy}. It is worth noting that for those functions mentioned above which satisfy the multifractal formalism, most of the time (see \cite{Jafsiam1,JafMand1,JAFFAAP,JAFFJMPA,BSw}) the wavelet expansion reveals that it is possible to closely relate the wavelet coefficients to the distribution of some positive Borel measure $\mu$ (sometimes discrete, as it can be shown for the saturation functions in Besov spaces \cite{JAFFJMPA}) satisfying the multifractal formalism for measures \cite{BMP,Olsen,Pesin}, for which the pointwise H\"older exponent is usually defined by 
\begin{equation}\label{expomu}
h_\mu(t)=\liminf_{r\to0^+} \frac{\log (\mu(B(t,r))}{\log (r)}.
\end{equation}

\smallskip

In practice, it may happen to be difficult to extract a good enough characterization of the sets $E_f(h)$ from the function $f$ expansion in wavelet series. This  leads to seeking for other methods of $h_f(t)$ estimation, or exponents that are close to $h_f(t)$ and easier to estimate. The most natural alternative is the first order oscillation exponent of $f$ defined as
$$
 h^{(1)}_f(t)= \liminf_{r\to 0^+} \frac{\log (\sup_{t,s\in B(t,r)}|f(s)-f(t)|)}{\log (r)}.
$$
If not an integer, $h^{(1)}_f(t)$ is equal to $h_f(t)$. When the function $f$ can be written as $g\circ \theta$, where $g$ is a monofractal function of (single) H\"older exponent $\gamma$ and $\mu=\theta'$ (the derivative of $\theta$ in the distributions sense) is a positive Borel measure satisfying the multifractal formalism for measures, we  have a convenient way to obtain the singularity spectrum of $f$ associated with the exponent $ h^{(1)}_f$ from that of $\mu$ (one exploits the equality $ h^{(1)}_f(t)=\gamma \,  h^{(1)}_\theta(t)=\gamma \, h_\mu(t)$ at good points $t$). Such a representation $f=g\circ \theta$ has been shown to exist for certain classes of deterministic multifractal functions mentioned above \cite{JafMand2,MandGaussian,Sadv}. 

\smallskip

It turns out that for the functions considered in this paper, in general wavelet basis expansions are not enough tractable to yield accurate information on the iso-H\"older sets. Also, this class of functions is versatile enough to contain elements which can be naturally represented under the same form $g\circ \theta$ as above, as well as elements for which such a natural decomposition does not exist. For these functions, inspired by the work achieved in \cite{JAFFJMP,Jafw}, we are going to compute the singularity spectrum by using the $m^{\tiny{th}}$ order oscillation pointwise exponents  ($m\ge 1$) and consider associated multifractal formalisms. To our best knowledge, this approach has not been used to treat a non-trivial example before.

\medskip

We denote by  $(f^{(m)})_{m\ge 1}$ the sequence of $f$ derivatives in the distribution sense. If $J$ is a non trivial compact subinterval of $I$, for $m\ge 1$, let
$$
{\rm Osc}^{(m)}_f(J)=\sup_{[t,t+mh]\subset J}|\Delta_h^mf(t)|,
$$
where
$\Delta_h^1f(t)=f(t+h)-f(t)$ and for $m\ge 2$, $\Delta_h^mf(t)=\Delta_h^{m-1}f(t+h)-\Delta_h^{m-1}f(t)$ (notice that ${\rm Osc}^{(1)}_f(J)=\sup_{s,t\in J}|f(s)-f(t)|$). Then, the pointwise oscillation exponent of order $m\ge 1$ of $f$ at $t\in\supp(f^{(m)})$ is defined as
$$
h^{(m)}_f(t)=\liminf_{r\to 0^+}\frac{\log {\rm Osc}^{(m)}_f(B(t,r))}{\log r}.
$$
We only  consider points in $\supp(f^{(m)})$, because from the pointwise regularity point of view, $\supp(f^{(m)})$ is the only set over which we can learn non-trivial information thanks to $h_f^{(m)}$. Indeed, outside this closed set, the function $f$ is locally equal to a polynomial of degree at most $m-1$, so $f$ is $C^\infty$.

The pointwise H\"older exponent $h_f$ carries  non-trivial information at points at which $f$ is not locally equal to a polynomial, that is points in $\bigcap_{m\ge 1} \supp(f^{(m)})$. 

If $t\in \bigcap_{m\ge 1} \supp(f^{(m)})$, it is clear that the sequence $(h^{(m)}_f(t))_{m\ge 1}$ is non decreasing. In fact, $\sup_{m\ge 1} h^{(m)}_f(t)=h_f(t)$. This is a consequence of Whitney's theorem on local approximation of functions by polynomial functions \cite{Whitney,Sendov} (the result is in fact proved for bounded functions): For every $m\ge 1$, there exists a constant $C_m$ (independent of $f$) such that for any subinterval $J$ of $I$, there exists a polynomial function $P$ of degree at most $m-1$ such that 
$$
|f(x)-P(x)|\le C_m {\rm Osc}^{(m)}_f(J).
$$
This, together with the definition of $h_f$ yields the following statement, which is also established in \cite{Jafw} by using wavelet expansion when $f$ is uniformly H\"older.
\begin{proposition}\label{hinfty}If $f:I\to\mathbb{C}$ is continuous, then for $t\in  \bigcap_{m\ge 1} \supp(f^{(m)})$, $h^{(m)}_f(t)$ converges to $h_f(t)$. Moreover, if $h_f(t)<\infty$, then $h^{(m)}_f(t)=h_f(t)$ for all $m> h_f(t)$. 
\end{proposition}

Now, the multifractal analysis of $f$ consists in computing singularity spectra like
\begin{equation}\label{specfunc}
h\ge 0\mapsto \dim_H E^{(m)}_f(h),
\end{equation}
where for $h\ge 0$ and $m\in\mathbb{N}_+$,
$$
E^{(m)}_f(h)=\big\{t\in \supp(f^{(m)}): h^{(m)}_f(t)=h\big \},
$$
and for $h\ge 0$, 
$$
E^{(\infty)}_f(h)=\Big\{t\in \bigcap_{n\ge 1} \supp(f^{(n)}): h^{(\infty)}_f(t)=h\Big \}, \quad \big(\text{where } h^{(\infty)}_f(t)=h_f(t)\big).
$$
Proposition~\ref{hinfty} yields 
$$E^{(\infty)}_f(h)=E^{(m)}_f(h) \quad (\forall h\ge 0,\ \forall \ m> h).
$$
Inspired by the multifractal formalisms for  measures on the line  \cite{Re,BMP,Olsen,Pesin,LN} as well as  multifractal formalism for functions in \cite{JAFFJMP,Jafw}, it is natural to consider for each $m\ge 1$ the $L^q$-spectrum of $f$ associated with the oscillations of order $m$, namely 
$$
\tau^{(m)}_{f}(q)=\liminf_{r\to 0}\frac{\log \sup \left\{\sum_i {\rm Osc}^{(m)}_f(B_i)^q\right \}}{\log (r)},
$$
where the supremum is taken over all the  families of  disjoint closed intervals $B_i$ of radius $r$ with centers in $\supp(f^{(m)})$.  For all $h\ge 0$ and $m\ge 1$, we have (Proposition~\ref{Formalism})
$$
\dim_H E^{(m)}_f(h)\le (\tau^{(m)}_f)^*(h)=\inf_{q\in\mathbb{R}}h q-\tau^{(m)}_f(q),
$$ 
and due to Proposition \ref{hinfty},
\begin{equation}\label{upperbound1}
\dim_H E^{(\infty)}_f(h)\le (\tau^{(\infty)}_f)^*(h):=\inf_{m>h}(\tau^{(m)}_f)^*(h),
\end{equation}
a negative dimension meaning that $E^{(m)}_f(h)$ is empty. We will say that {\it the multifractal formalism holds for $f$ and $m\in \mathbb{N}_+\cup\{\infty\}$ at $h\ge 0$ if $E^{(m)}_f(h)$ is not empty and $\dim_H E^{(m)}_f(h)=(\tau^{(m)}_f)^*(h)$}.  

When $m=\infty$, the exponent $h^{(m)}_f$ is naturally stable by addition of a $C^\infty$ function, and so is the validity of the associated multifractal formalism. This is not the case when $m<\infty$ (see Corollary~\ref{CW} for an illustration).

\smallskip

As we said, our approach for the multifractal formalism is inspired by the ``oscillation method" introduced in \cite{JAFFJMP,Jafw} for uniformly H\"older functions. There, quantities like  $\tau^{(m)}_f$ are computed by using balls centered at points of finer and finer regular grids, and only for $q\ge 0$. So our definition of $\tau^{(m)}_f$ is more intrinsic, though equivalent. The choice $q\ge 0$ in \cite{Jafw} corresponds to the introduction of some functions spaces related with the functions $\tau^{(m)}_f$ that provide a natural link between  wavelets and oscillations approach to the multifractal formalism when $q\ge 0$ and $f$ is uniformly H\"older. It is worth noting that thanks to this link, for any $q\ge 0$, if we define $n_q$ as the smallest integer $n$ such that $n q-1\ge \tau_f^{(n)}(q)$, then for all $n\ge n_q$, the function $\tau^{(n)}_f$ coincides on the interval $[q,\infty]$ with the {\it scaling function} $\tau^{\mathit{W}}_f$ associated with the so-called {\it wavelet leaders} in \cite{JAFFJMP,Jafw}. This implies that for $h\ge 0$ such that the multifractal formalism holds at $h$ for $m=\infty$, even though  $E^{(n)}_f(h)=E^{(\infty)}_f(h)$ for all $n\ge [h]+1$,  $\dim_H E^{(\infty)}_f(h)$ may be equal to $(\tau^{(n)}_f)^*(h)$ only for $n\gg[h]+1$ as $h$ tends to 0.

\medskip

We now introduce the functions whose multifractal analysis will be achieved in this paper. 

\medskip

We fix  an integer $b\geq 2$. For every $n\ge 0$ we define $\mathscr{A}^n=\{0,\dots,b-1\}^n$ (by convention $\mathscr{A}^0$ contains the emty word denoted $\emptyset$), $\mathscr{A}^*=\bigcup_{n\ge 0} \mathscr{A}^n$, and $\mathscr{A}^{\mathbb{N}_+}=\{0,\dots,b-1\}^{\mathbb{N}_+}$.

If $n\ge 1$, and $w=w_1\cdots w_n\in \mathscr{A}^n$ then for every $1\le k\le n$, the word $w_1\dots w_k$ is denoted $w|_{k}$, and if $k=0$ then $w|_0$ stands for $\emptyset$.  Also, if $t\in\mathscr{A}^{\mathbb{N}_+}$ and $n\ge 1$, $t|_{n}$ denotes the word $t_1\cdots t_n$ and $t|_0$ the empty word. 

We denote by $\pi$ the natural projection of $\mathscr{A}^{\mathbb{N}_+}$ onto $[0,1]$: If $t\in \mathscr{A}^{\mathbb{N}_+}$, $\pi (t)=\sum_{k=1}^\infty t_kb^{-k}$.

When $t\in [0,1]$ is not a $b$-adic point, we identify it with the element of $\mathscr{A}^{\mathbb{N}_+}$  which represent its $b$-adic expansion, namely the element of $\pi^{-1}(\{t\})$.

We consider a sequence of independent copies $(W(w))_{w\in \mathscr{A}^*}$ of a random vector $$W=(W_0,\dots,W_{b-1})$$ whose components are complex, integrable, and satisfy $\mathbb{E}(\sum_{i=0}^{b-1}W_i)=1$. Then,  we define the sequence of functions 
\begin{equation}\label{Fn'}
F_{W,n}(t)=\int_0^t b^{n}\prod_{k=1}^{n}W_{u_{k}}(u|_{k-1})\, \text{d}u.
\end{equation}

For $q\in\mathbb{R}$ let
\begin{equation}\label{varphi}
 \varphi_W(q)=-\log_b\mathbb{E}\Big (\sum_{i=0}^{b-1}\big (\mathbf{1}_{\{W_i\neq 0\}}|W_i|^q\big )\Big ).
\end{equation}
The assumption $\mathbb{E}(\sum_{i=0}^{b-1}W_i)=1$ implies that $\varphi_W(1)\le 0$ with equality if and only if $W\ge 0$, i.e., the components of $W$ are non-negative almost surely. In this case only, all the functions $F_{W,n}$ are non-decreasing almost surely.

\smallskip

The following results are established in \cite{BJMpartII}. 

\medskip

\noindent
{\bf Theorem\,A} \cite{BJMpartII} {\bf (Non-conservative case)} {\it Suppose that $\mathbb{P}\Big (\sum_{i=0}^{b-1}W_i\neq 1\Big )>0$ and  there exists $p>1$ such that $\varphi_W(p)>0$. Suppose, moreover, that either $p\in (1,2]$ or $\varphi_W(2)>0$.

\begin{enumerate}
\item
$(F_n)_{n\ge 1}$ converges uniformly, almost surely and in $L^p$ norm, as $n$ tends to $\infty$, to a function $F=F_W$, which is non decreasing if $W\ge 0$.  Moreover, the function $F$ is $\gamma$-H\"older continuous for all $\gamma$ in $(0,\max_{q\in(1,p]}\varphi_W(q)/q)$.

\item $F$ satisfies the statistical scaling invariance property:
\begin{equation}\label{sesi}
F=\sum_{i=0}^{b-1}\mathbf{1}_{[i/b,(i+1)/b]}\cdot \Big (F(i/b)+W_i\, F_i\circ S_i^{-1}\Big ),
\end{equation}
where $S_i(t)=(t+i)/b$, the random objects $W$, $F_0,\dots,F_{b-1}$ are independent, and the $F_i$ are distributed like $F$ and the equality holds almost surely.
\end{enumerate}
}

\medskip

\noindent
{\bf Theorem\,B} \cite{BJMpartII} {\bf(Conservative case)} {\it Suppose that $\mathbb{P}\Big (\sum_{i=0}^{b-1}W_i=1\Big )=1$. 

\begin{enumerate}
\item
If there exists $p>1$ such that $\varphi_W(p)>0$, then the same conclusions as in Theorem~A hold. 
\smallskip

\item {\bf (Critical case)} Suppose that $\lim_{p\to\infty} \varphi_W(p)=0$ (in particular $\varphi_W$ is increasing and $\varphi_W(p)<0$ for all $p>1$). This is equivalent to the fact that $\mathbb{P}(\forall\ 0\le i\le b-1,\ |W_i|\le 1)=1$ and $\sum_{i=0}^{b-1}\mathbb{P}(|W_i|=1)=1$.

\smallskip

Suppose also that  $\mathbb{P}(\#\{i: |W_i|=1\}=1)<1$,  and there exists $\gamma\in (0,1)$ such that, with probability 1, one of the two following properties holds for each $0\le i\le b-1$
\begin{equation}\label{critical}
\begin{cases}
\text{either }  |W_i|\le \gamma,\\
\text{or } |W_i|=1 \text{ and } \Big (\sum_{k=0}^{i-1} W_i,\sum_{k=0}^i W_i\Big ) \in \{(0,1),(1,0)\} 
\end{cases}.
\end{equation}
Then, with probability 1, $(F_n)_{n\ge 1}$ converges almost surely uniformly to a limit $F=F_W$ which is  nowhere locally uniformly H\"older and satisfies part 2. of Theorem~A. 
\end{enumerate}
}
When the components of $W$ are non-negative (resp. positive), the function $F_W$ is non-decreasing (resp. increasing) and the measure $F_W'$ is the measure considered in \cite{M2,KP}.

\medskip

In the rest of the paper, we will work with the natural and more general model of function constructed as follows. Instead of considering only one multiplicative cascade, we consider a couple $(W,L)$ of random vectors taking values in $\mathbb{C}^b\times {\mathbb{R}_+^*}^b$. We assume that both $W$ and $L$ satisfy the same property as $W$ in the previous paragraph:  $\mathbb{E}(\sum_{i=0}^{b-1}W_i)=1=\mathbb{E}(\sum_{i=0}^{b-1}L_i)$.

We consider a sequence of independent copies $(W(w),L(w))_{w\in \mathscr{A}^*}$ of $(W,L)$, and we also assume that both $W$ and $L$ satisfy the assumptions of Theorem~A or B. This yields almost surely two continuous, functions $F_W$ and $F_L$, the former being increasing. The function we consider over $[0,F_L(1)]$ is 
$$
F=F_W\circ F_L^{-1}.
$$ 

When $F_W$ is non-decreasing, the measure $F'$ has been considered in \cite{B2}, and also in \cite{AP} under the assumption that  $\sum_{i=0}^{b-1}L_i=1$ almost surely. 

If the components of $W$ and $L$ are deterministic real numbers and $\sum_{i=0}^{b-1}W_i=1=\sum_{i=0}^{b-1}L_i$, we recover the self-affine functions constructed in \cite{Bed}. The multifractal analysis of these functions has been achieved in \cite{Jafsiam1} by using their wavelet expansion (however, the endpoints of the spectrum are not investigated).  It is also possible to use the alternative approach consisting in showing that $F_W$ can be represented as a monofractal functions in multifractal time \cite{MandGaussian,Sadv}, and then consider the exponent $h_F^{(1)}$ rather than $h_F$. It turns out that such a time change also exists in the random case under restrictive assumptions on $W$, which include the deterministic case (see \cite{BJMpartII}). This is useful because, as we said, our calculations showed that in general in the random case it seems difficult to exploit the wavelet transform of $F$ to compute its singularity spectrum. Moreover, this approach could not cover all the cases since for the functions build in Theorem~B(2), there is no natural time change (see \cite{BJMpartII}). Also, these functions are nowhere locally uniformly H\"older and do not belong to any critical Besov space (specifically, their singularity spectra have an infinite slope at 0), so that there is few expectation to  characterize their pointwise H\"older exponents through their wavelet transforms. 

Using the $m^{\tiny{th}}$ order oscillation pointwise exponents provides an efficient alternative tool. We obtain the following results (for simplicity, we postpone to Section~\ref{waekassump} the discussion of an extension under weaker assumptions). We discard the obvious case where $W=L$, for which $F=\text{Id}_{[0,F_L(1)]}$ almost surely. Also, we assume that 
$$
\text{$\varphi_L>-\infty$ over $\mathbb{R}$ and $0<L_i< 1$ almost surely. }
$$

The first result concerns functions $F$ with bell-shaped singularity spectra. We find that for some of these functions, the left endpoint of their spectra is equal to $0$. This is a new phenomenon in the multifractal analysis of statistically self-similar continuous functions.      

\begin{theorem}{\bf (Bell shaped spectra)}\label{AM} Suppose that $\mathbb{P}(\sum_{i=0}^{b-1}\mathbf{1}_{\{W_i\neq 0\}}\ge 2)=1$ and $\varphi_W>-\infty$  over $\mathbb{R}$. For $q\in\mathbb{R}$, let $\tau(q)$ be the unique solution of the equation $\mathbb{E}(\sum_{i=0}^{b-1}\mathbf{1}_{\{W_i\neq 0\}}|W_i|^qL_i^{-t})=1$. The function $\tau$ is concave and analytic. With probability 1, 

\begin{enumerate}
\item  $\supp(F^{(m)})=\supp(F')$ for all $m\in\mathbb{N}_+$ and $\dim_H \supp(F')=-\tau(0)$. 

\item For all $h\ge 0$ and $m\in \mathbb{N}_+\cup\{\infty\}$, $\dim_H E^{(m)}_F(h)=(\tau^{(m)}_F)^*(h)=(\tau^{(1)}_F)^*(h)$, a negative dimension meaning that $E^{(m)}_F(h)$ is empty. Moreover, $E^{(m)}_F(h)\neq\emptyset$ if $(\tau^{(1)}_F)^*(h)=0$. In other words, for all $m\in \mathbb{N}_+\cup\{\infty\},$ $F$ obeys the multifractal formalism  at every $h\ge 0$ such that $(\tau^{(m)}_F)^*(h)\ge 0$. In addition, if $F$ is built as in Theorem~B(2) (critical case), the left endpoint of these singularity spectra is the exponent $0$, and the corresponding level set is dense, with Hausdorff dimension 0. 

\item For all $m\in\mathbb{N}_+$, $\tau^{(m)}_F=\tau$ on the interval $J=\{q\in\mathbb{R}:\tau'(q)q-\tau(q)\ge 0\}$, and if $\overline{q}=\sup(J)<\infty$ (resp. $\underline{q}:=\inf (J)>-\infty$) then $\tau^{(m)}_F(q)=\tau'(\overline{q})q$ (resp. $\tau'(\underline{q})q$) over $[\overline{q},\infty)$ (resp. $(-\infty,\underline{q}]$). \\
Moreover, if  there does not exist $H\in (0,1)$ such that for all $0\le i\le b-1$ we have $|W_i|\in \{0,L_i^H\}$ then $\tau$ is strictly concave over $J$; otherwise, $\tau(q)=qH+\tau(0)$ and $F$ is monofractal with a H\"older exponent equal to $H$. 
\end{enumerate}
\end{theorem}

\begin{figure}[htp]
\centering
\begin{minipage}[c]{\textwidth}
\centering
\includegraphics[width=0.45\textwidth]{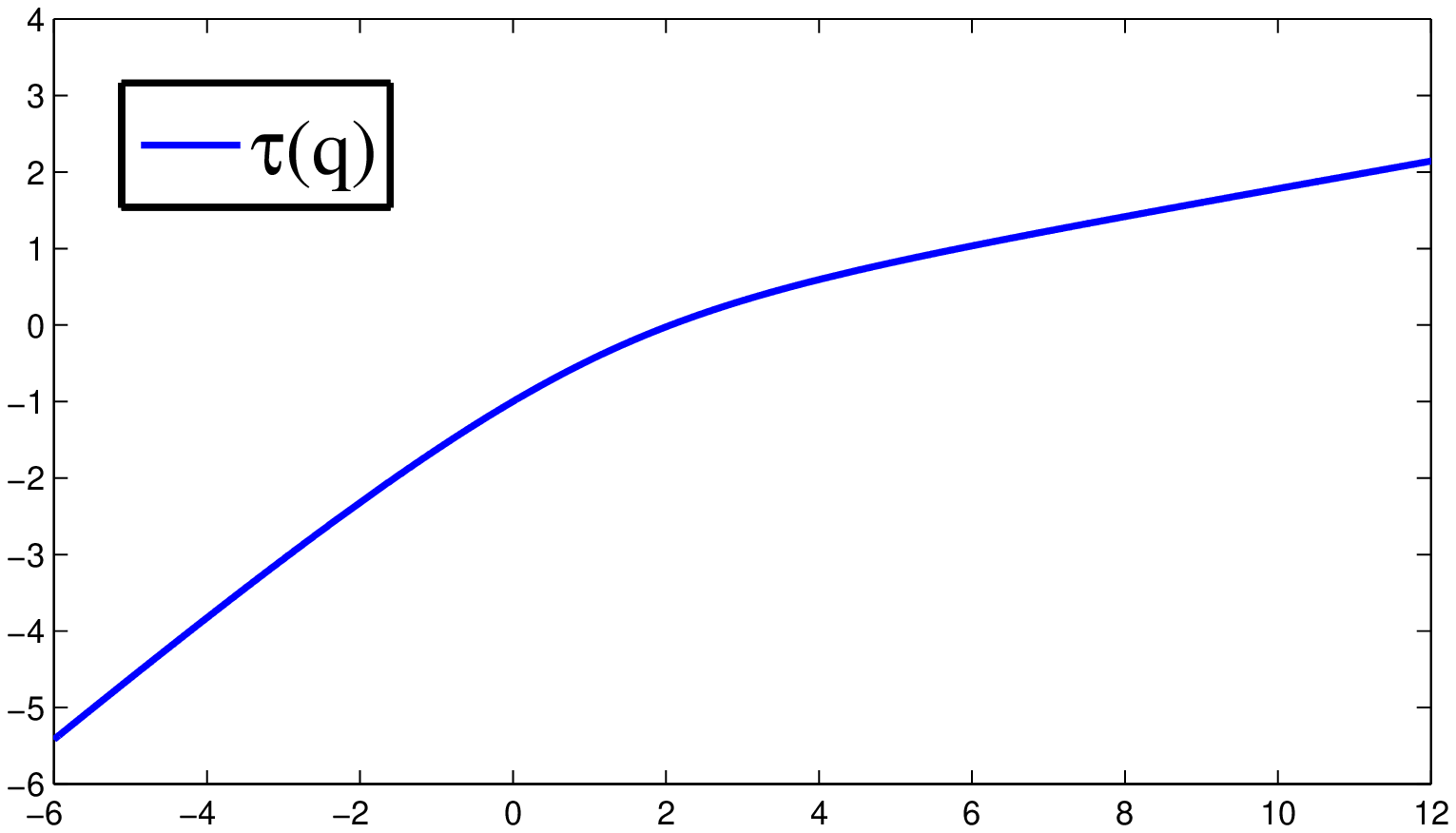}%
\includegraphics[width=0.45\textwidth]{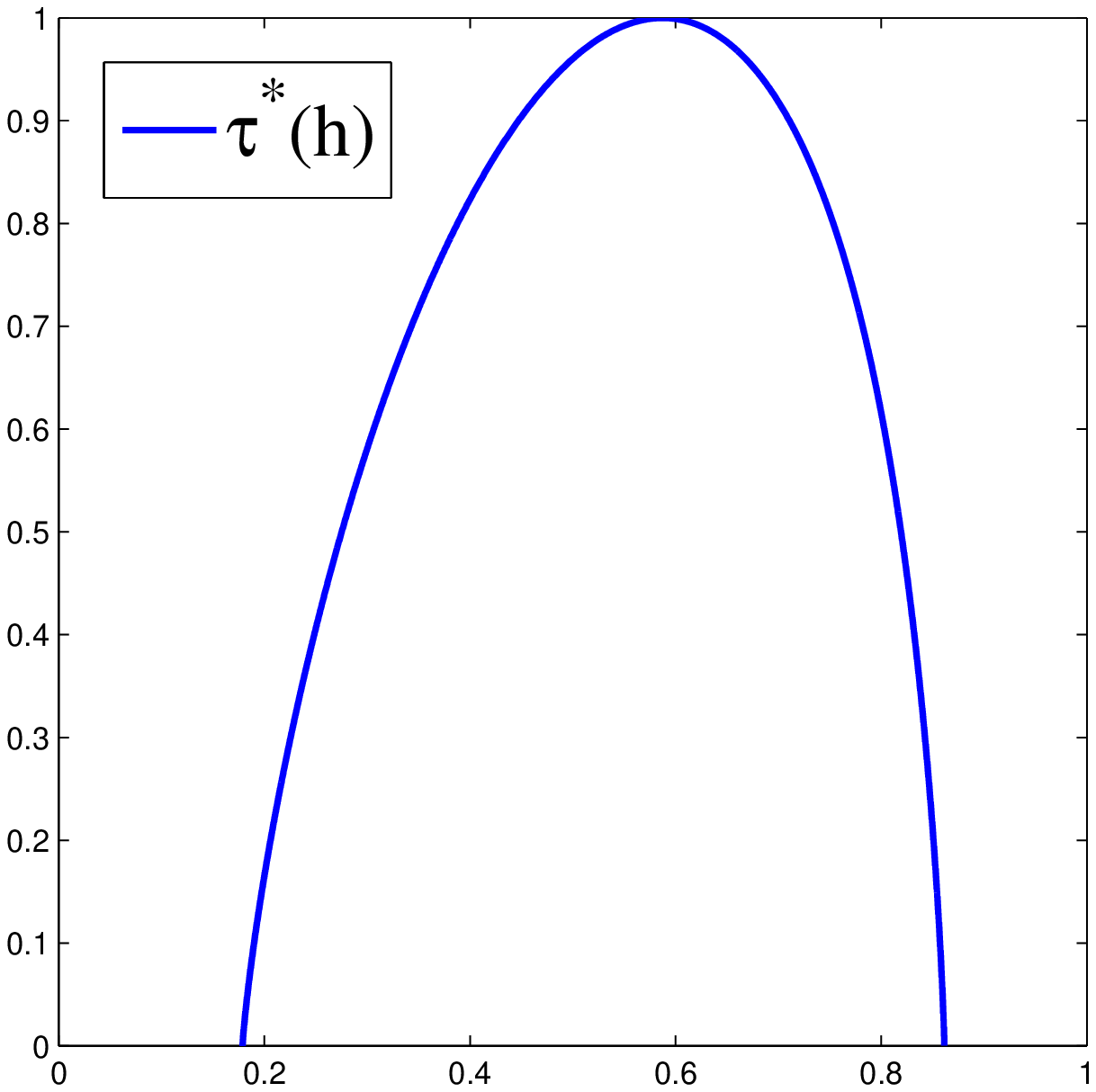}
\caption{Bell shaped spectrum in the case where the left endpoint is not $0$.}\label{bell}
\end{minipage}
\begin{minipage}[c]{\textwidth}
\centering
\includegraphics[width=0.45\textwidth]{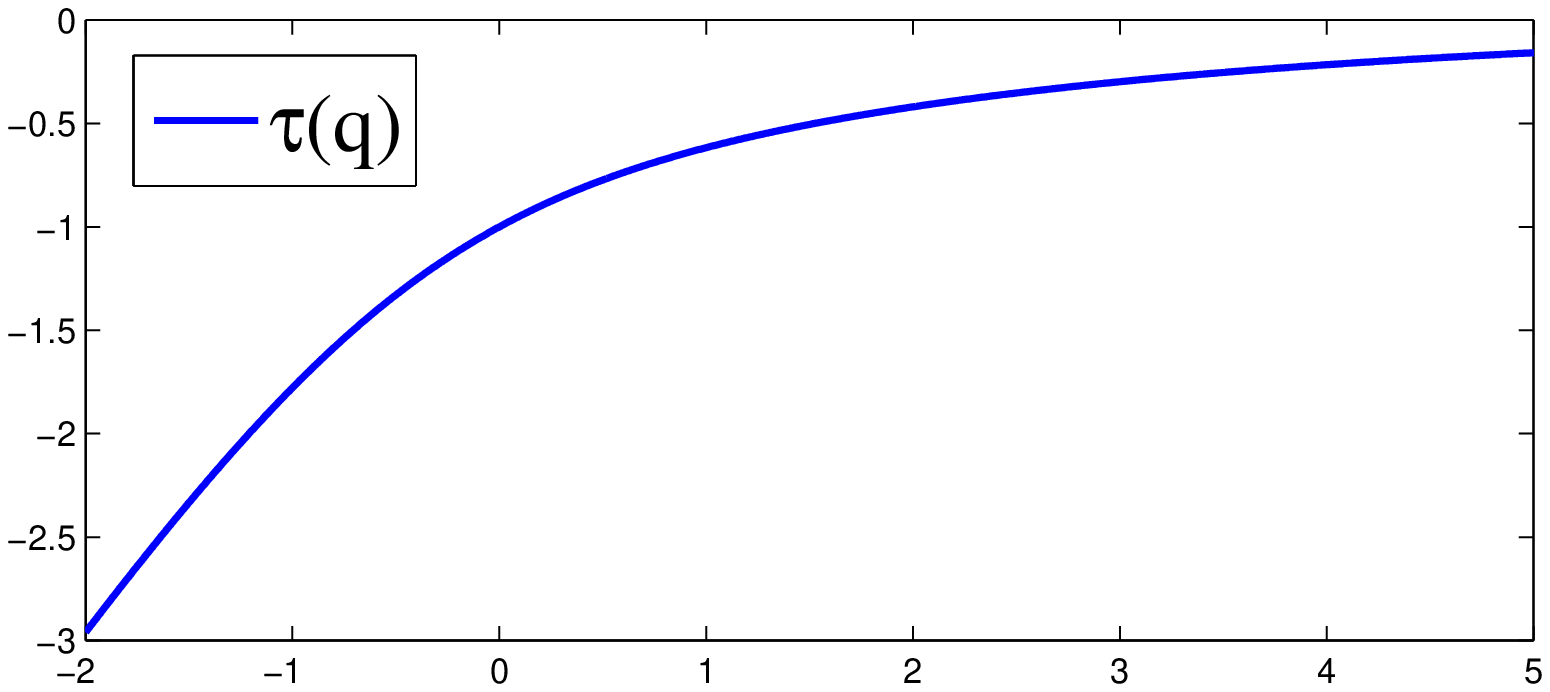}%
\includegraphics[width=0.45\textwidth]{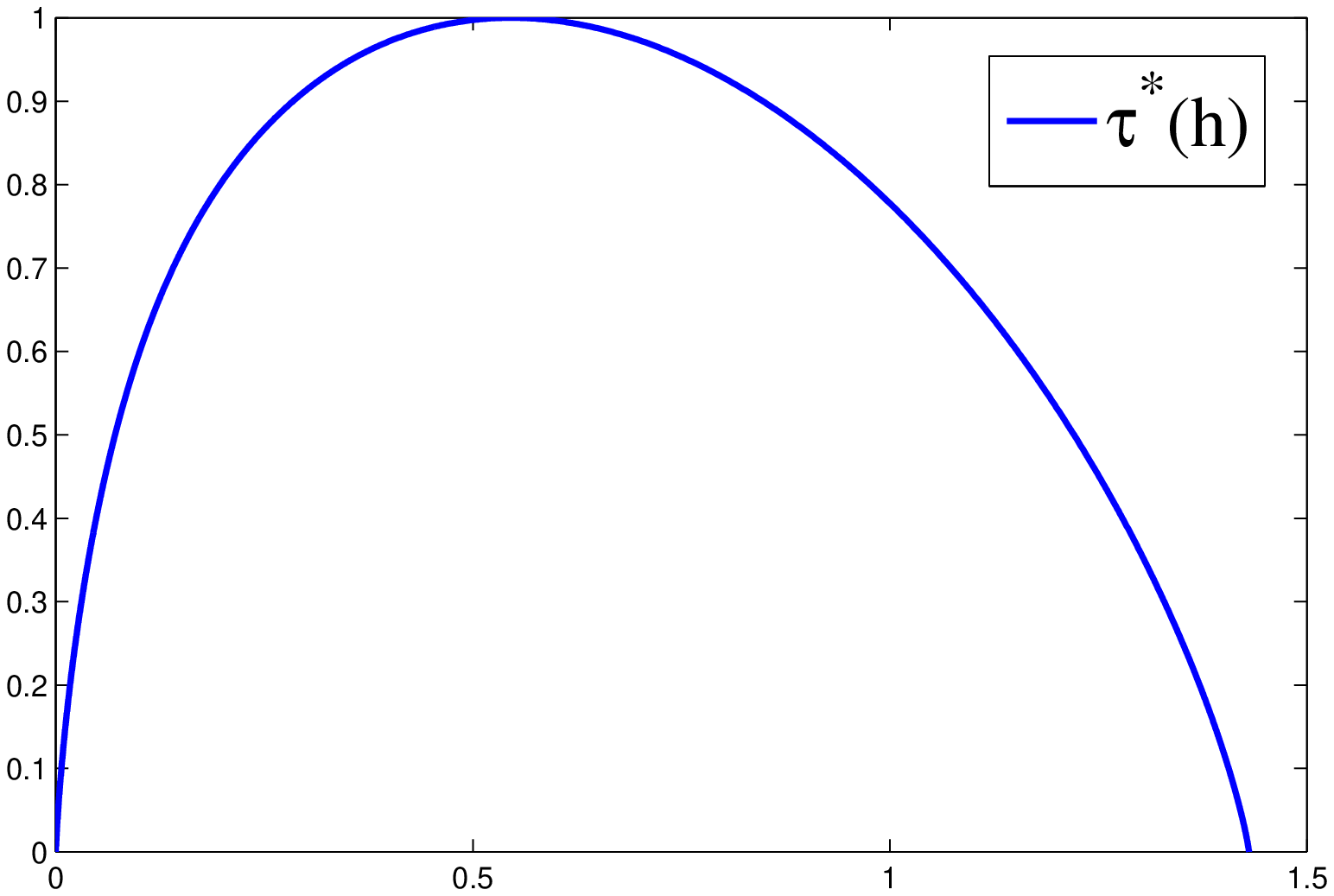}
\caption{Bell shaped spectrum in the critical case where the left endpoint is $0$.}\label{bell1}
\end{minipage}
\end{figure}

Notice that $-\tau(0)<1$ if and only if at least one component of $W$ vanishes with positive probability, and in this case the support of $F'$ is a Cantor set.

 In the next result, we get functions $F$ obeying the multifractal formalism and for which the singularity spectra are left-sided, i.e., increasing,  and with a support equal to the whole interval $[0,\infty]$. This is another new phenomenon in multifractal analysis of continuous staitistically self-similar functions. 
\begin{theorem}\label{AM2} {\bf (Left-sided spectra)} Suppose that $\mathbb{P}(\sum_{i=0}^{b-1}\mathbf{1}_{\{W_i\neq 0\}}\ge 2)=1$ and $\varphi_W(q)>-\infty$ over $\mathbb{R}_+$. For $q\in\mathbb{R}_+$, let $\tau(q)$ be defined as in Theorem~\ref{AM} . The function $\tau$ is concave, and analytic over $(0,\infty)$. \\ Suppose also that $\mathbb{E}(\sum_{i=0}^{b-1}\mathbf{1}_{\{W_i\neq 0\}}L_i\log(|W_i|))=-\infty$, i.e. $\tau'(0)=\infty$.  Finally, suppose that $\mathbb{E}\big ((\max_{0\le i\le b-1} |W_i|)^{-\varepsilon}\big )<\infty$ for some $\varepsilon>0$. 

Then, the same conclusions as in Theorem~\ref{AM} hold. Moreover, the singularity spectra are left-sided, and $h^{(m)}_F=\infty$ for all $m\in\mathbb{N}_+\cup\{\infty\}$ on a set of full dimension in $\supp(F')$. In addition, if $F$ is built as in Theorem~B(2) (critical case), the support of the spectra is $[0,\infty]$. 
\end{theorem}

\begin{figure}[htp]
\centering
\begin{minipage}[c]{\textwidth}
\centering
\includegraphics[width=0.8\textwidth]{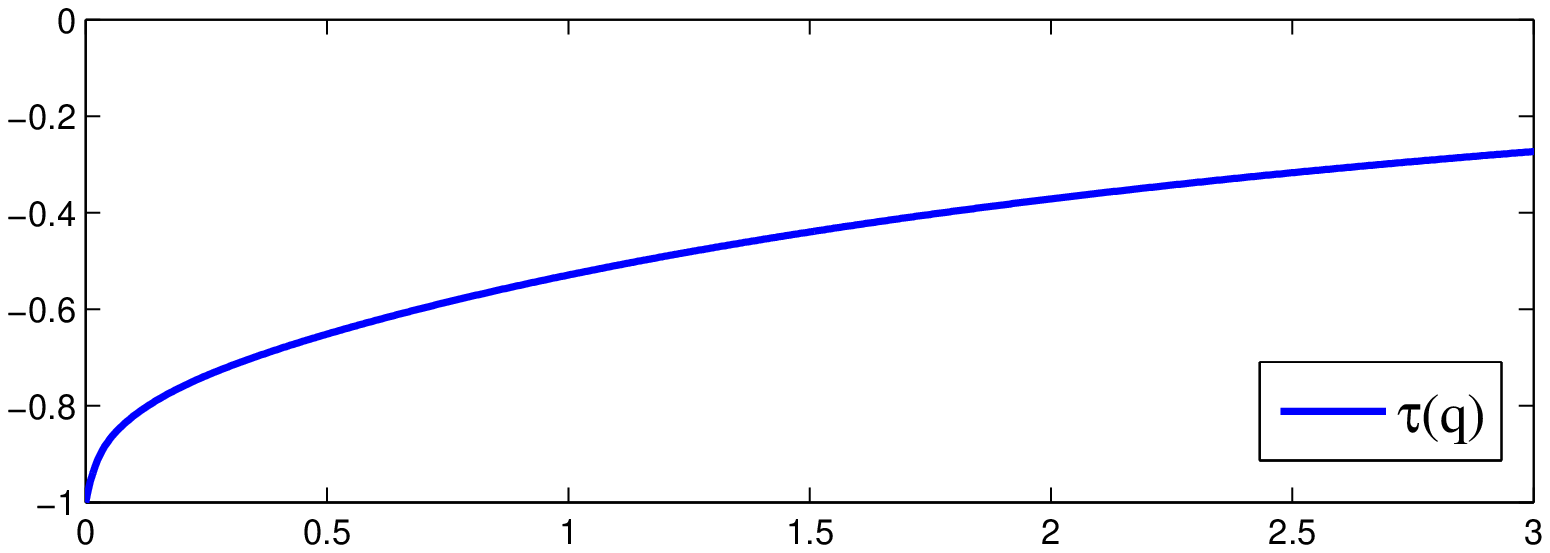}
\end{minipage}
\begin{minipage}[c]{\textwidth}
\centering
\includegraphics[width=0.8\textwidth]{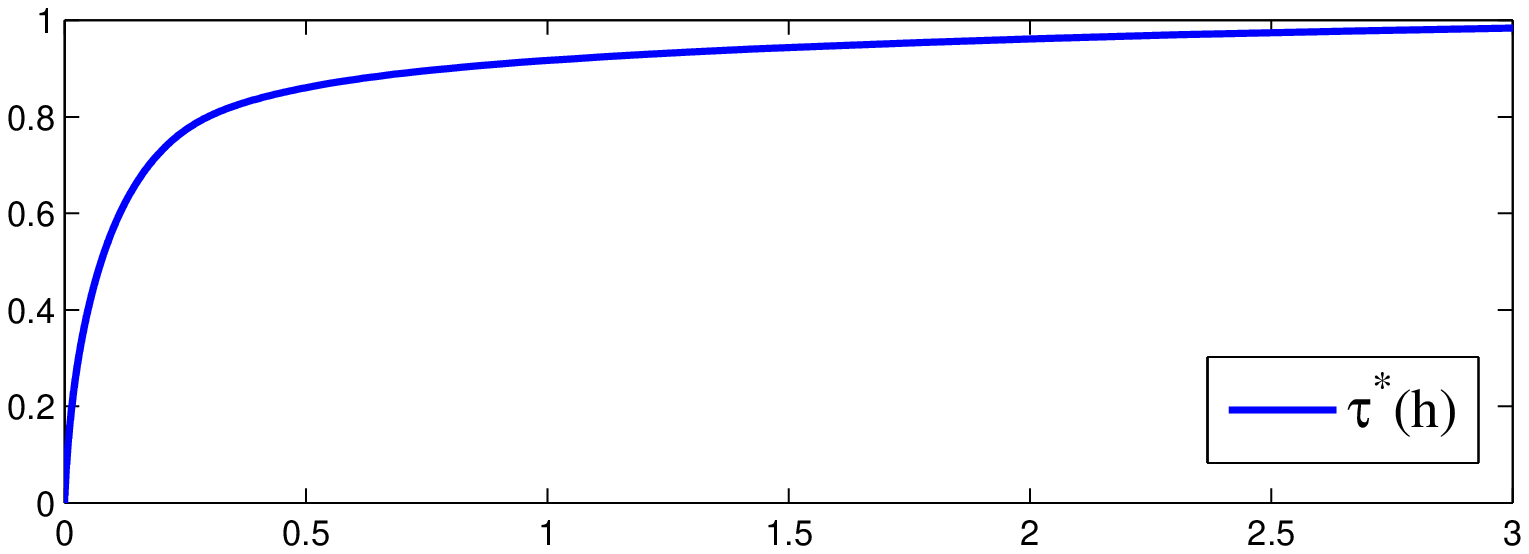}
\end{minipage}
\caption{Concave left-sided spectra with support $[0,\infty]$ in the critical case.}\label{left}
\end{figure}

\begin{remark}
\label{endpoint}
Examples of left sided spectra do exist for some other (increasing) continuous functions over $[0,1]$ possessing self-similarity properties \cite{M3,MR,MEH}, but their spectra do not contain the left endpoint 0.  

It is also worth mentioning that in some Besov spaces of continuous functions, the generic singularity spectrum is left sided, supported by a compact interval, and linear; moreover, the left-end point of this spectrum is equal to 0 for critical Besov spaces \cite{JAFFJMPA,JaffMeyer}. 

In the critical case considered in this paper (Theorem B(2)), the slope of the singularity spectra at 0 is equal to $\infty$ because of the duality between $h=\tau'(q)$ and $q=(\tau^*)'(h)$, and $h\to 0$ corresponds to $q\to\infty$.  
\end{remark}

\begin{remark} In the non-decreasing case (the components of $W$ are nonnegative), results on the multifractal analysis of the measure $\mu=F'$ have been obtained in several papers (which also deal with measures on $\mathbb{R}^d$). For the one dimensional case we are dealing with, the previous statements are substantial improvements of these results for the following reasons. 

At first, all these works only consider the first order oscillation exponent, which is sometimes computed only on the ``distorded" grid associated with the increments of $F_L$ as described above \cite{HoWa,Mol,B2}, and not in the more intrinsic way \eqref{expomu}. Moreover, in the papers which deal with the intrinsic exponent $h_\mu$, the assumptions on $W$ and $L$ are very strong: Their components must be bounded away from 0 and 1 by positive constants, and their sum must be equal to 1 almost surely \cite{AP,Fal}; moreover the result holds only for all $h\ge 0$ such that $\tau^*(h)>0$ almost surely, and not almost surely for all $h\ge 0$ such that  $\tau^*(h)> 0$. Also, the case of left sided spectra is not treated in these papers. 

Another important improvement concerns the computation of the endpoints of the singularity spectrum, which is a delicate issue; indeed it is already non-trivial to prove that the corresponding iso-H\"older sets are not empty. Our result includes the description of these endpoints, i.e. the endpoints of ${\tau_F^*}^{-1}(\mathbb{R}_+)$, without restriction on the behavior of $\tau$. This is a progress with respect to the work achieved in \cite{B2} where the case when $\overline{q}=\infty$ (resp. $\underline{q}=-\infty$ and $\lim_{q\to\infty}{\rm (resp.}\  \lim_{q\to-\infty} {\rm )} \tau'(q)q-\tau(q)=0$ was not worked out (in the present paper this is particularly important in the critical case of Theorem~B(2)), and where the H\"older exponents are computed only on  the grid naturally defined by $F_L$. Also, the new method we introduce to study the endpoints could be used to deal with the same question for the general class of random measures considered in \cite{BMS2}.\end{remark}

\begin{remark}
In the previous results, all the formalisms yield the same information.  In particular our discussion on the link between the oscillations and wavelets methods developed in \cite{Jafw} shows that when $F$ is uniformly H\"older, the multifractal formalism using wavelets also holds for such a function in the increasing part of the spectrum, without it be necessary to compute any wavelet transform. 
\end{remark}

The next result illustrates the unstability of the exponents and spectra associated with the $m^{\text{\tiny th}}$ order oscillations by addition of a $C^\infty$ function. 

\begin{corollary}\label{CW}
Let $f$ be a complex valued $C^\infty$ function over $\mathbb{R}_+$ such that for all $m\in\mathbb{N}_+$ the function $f^{(m)}$ does not vanish. Let $F$ be as in Theorem~\ref{AM2} and let $G=F+f$. The functions $F$ and $G$ have the same multifractal behavior from the pointwise H\"older exponent point of view.

For $m\in\mathbb{N}_+$, let $q_m$ be the unique real number such that $\tau(q_m)=q_mm-1$. 

With probability 1, for all $m\in\mathbb{N}_+$, we have $\tau_G^{(m)}=\tau_F^{(m)}=\tau$ over $[q_m,\infty)$, and $\tau_G^{(m)}(q)=qm-1$ for $0\le q<q_m$. Moreover, for all $m\in\mathbb{N}_+$, the multifractal formalism holds at every $h\in[0,\tau'(q_m)]$ such that ${\tau_G^{(m)}}^*(h)\ge 0$ as well as at $h=m$, and for all $h\in (\tau'(q_m),m)$ we have  $\dim_H\, E^{(m)}_G(h)=\tau^*(h)<(\tau_G^{(m)})^*(h)$. 
\label{addf}
\end{corollary}

\begin{figure}[htb]\label{left1}
\centering
\begin{minipage}[c]{\textwidth}
\centering
\includegraphics[width=0.8\textwidth]{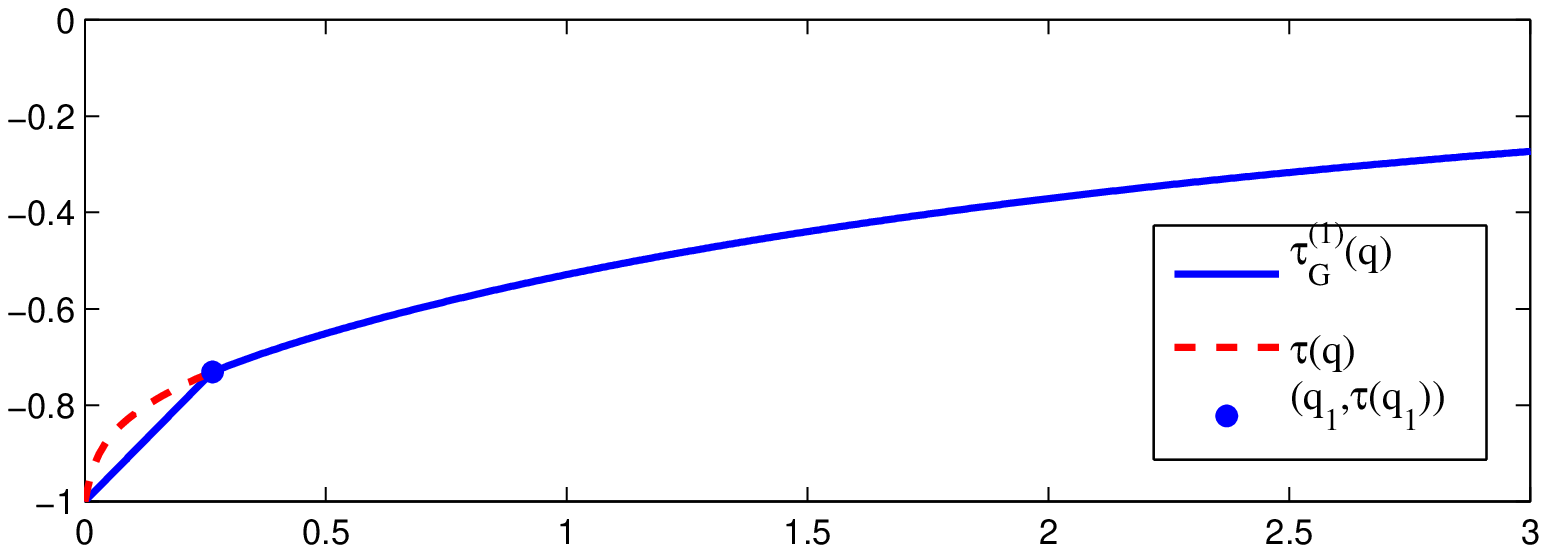}
\end{minipage}
\begin{minipage}[c]{\textwidth}
\centering
\includegraphics[width=0.8\textwidth]{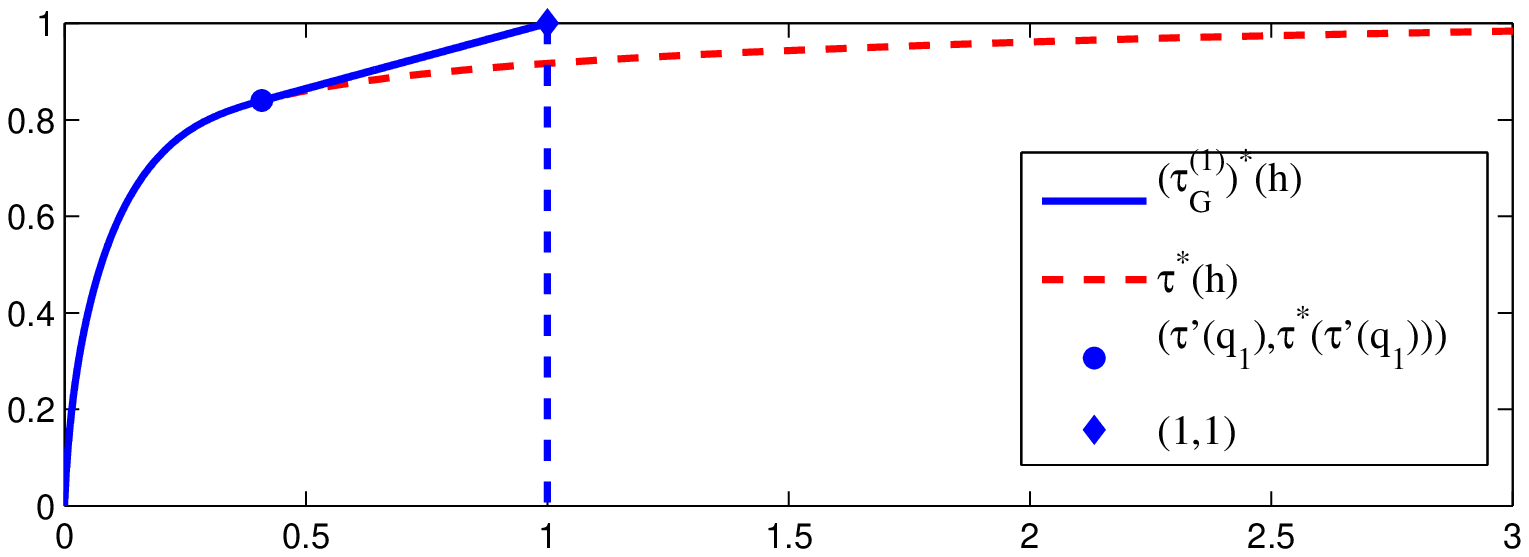}
\end{minipage}
\caption{Top: $\tau_G^{(1)}(q)=\min\{q-1,\tau(q)\}$ for $q\ge 0$. Bottom: $(\tau_G^{(1)})^*(h)=\tau^*(h)$ for $h\in [0,\tau'(q_1)]$, $(\tau_G^{(1)})^*(h)=\tau^*(\tau'(q_1))+q_1(h-\tau'(q_1))$ for $h\in (\tau'(q_1),1]$ and $(\tau_G^{(1)})^*(h)=1$ elsewhere.}
\end{figure}

We end this section with additional definitions. 

\medskip

\noindent
{\bf Definitions.} 
\medskip

\noindent
{\it The coding space.}
\medskip

The word obtained by concatenation of $u\in\mathscr{A}^*$ and $v\in\mathscr{A}^*\cup\mathscr{A}^{\mathbb{N}_+}$ is denoted $u\cdot v$ and sometimes $uv$. 
For every $w\in \mathscr{A}^*$, the cylinder with root $w$, i.e. $\{w\cdot t:t\in \mathscr{A}^{\mathbb{N}_+}\}$ is denoted $[w]$. The $\sigma$-algebra generated in $\mathscr{A}^{\mathbb{N}_+}$ by the cylinders, namely $\sigma([w]:w\in\mathscr{A}^*)$ is denoted $\mathcal{S}$. The set $\mathscr{A}^{\mathbb{N}_+}$ is endowed with the standard metric distance $d(t,s)=\inf\{b^{-n}: n\ge 0,\  \exists \  w\in\mathscr{A}^n,\ t,s\in [w]\}$. Then  the Borel $\sigma$-algebra is equal to $\mathcal{S}$. 

For every $n\ge 0$, the length of an element of $\mathscr{A}^n$ is by definition equal to $n$ and we denote it $|w|$. For $w \in \mathscr{A}^*$, we define $I_w=[t_w,t_w+b^{-|w|})$ and $I^L_w=F_L(I_w)$. We denote by $w^-$ or $w^{-1}$ (resp.  $w^+$ or $w^{+1}$)  the unique element of $\mathscr{A}^{|w|}$ such that $t_w-t_{w^-}=b^{-|w|}$ (resp. $t_{w^+}-t_w=b^{-|w|}$)  whenever $t_w\neq 0$ (resp. $t_w\neq 1-b^{-|w|}$) . We also denote $w$ by $w^0$. 

\medskip

\noindent
{\it Independent copies of $F_W$ and $F_L$, and associated quantities.}
\medskip

If $w\in\mathscr{A}^*$, $n\ge 1$ and $U\in\{W,L\}$, we denote by $F^{[w]}_{U,n}$ the function constructed as $F_{U,n}$, but with the weights $(U(w\cdot v))_{v\in\mathscr{A}^*}$. By construction, $F^{[\emptyset]}_{U,n}=F_{U,n}$, and 
$$
F^{[w]}_{U,n}(t)=\int_0^t b^{n}\prod_{k=1}^{n}U_{u_{k}}(w\cdot u| \, |w|+k-1)\, \text{d}u.
$$
We denote by $F_U^{[w]}$ the almost sure uniform limit of $(F^{[w]}_{U,n})_{n\ge 1}$.  We also define
$$
Q_U(w)=\prod_{k=1}^{n}U_{w_{k}}(w|_{k-1}).
$$

For $m\ge 1$ we denote ${\rm Osc}^{(m)}_{F_U}([0,1])$ by $Z_U^{(m)}$ and more generally ${\rm Osc}^{(m)}_{F_U^{[w]}}([0,1])$ by $Z_U^{(m)}(w)$.  Also, we denote ${\rm Osc}^{(m)}_{F_U}(I_w)$ by ${\rm O}^{(m)}_U(w)$.  By construction,  we have 
\begin{eqnarray}
\label{self-sim2}
{\rm Osc}^{(m)}_F(I^L_w)&=&{\rm Osc}^{(m)}_{F_W}(I_w)={\rm O}^{(m)}_W(w)=|Q_W(w)|Z_W^{(m)}(w),\\
\label{self-sim2'}
|I^L_w|&=&{\rm Osc}^{(1)}_{F_L}(I_w)={\rm O}^{(1)}_L(w)=Q_L(w)Z_L^{(1)}(w).
\end{eqnarray}

For $(q,t)\in\mathbb{R}^2$ let 
\begin{equation}\label{Phi}
\Phi(q,t)=\mathbb{E}\Big (\sum_{i=0}^{b-1}\mathbf{1}_{\{W_i\neq 0\}}|W_i|^qL_i^{-t}\Big )\text{ and } \Psi(q,t)=\mathbb{E}\big ({\rm Osc}_{F_w}([0,1])^q{F_L}(1)^{-t}\big ).
\end{equation}

\medskip

\noindent
{\it Hausdorff dimension.}
\medskip

If $(X,d)$ is a locally compact metric space, for $D\in\mathbb{R}$, $\delta>0$, and $E\subset X$, let 
$$
\mathcal{H}^D_\delta(E)=\inf\{\sum_{i\in I} |U_i|^D\},
$$
where the infimum is taken over the set of all the at most countable coverings $\bigcup_{i\in I}U_i$ of $E$ such that $0\le |U_i|\le \delta$, where $|U_i|$ stands for the diameter of $U_i$ and by convention $0^D=0$.  Then define 
$$
\mathcal{H}^D(E)=\lim_{\delta\searrow 0} \mathcal{H}^D_\delta(E)
$$
($\mathcal{H}^D_\delta(E)$ is by construction a non-increasing function of $\delta$). If $D\ge 0$, $\mathcal{H}^D(E)$ is called the $D$-dimensional Hausdorff measure of $E$. The Hausdorff dimension of $E$ is the number 
$$
\dim_H E=\inf\{D:\mathcal{H}^D(E)<\infty\}.
$$
It is clear that we have $\dim_H E<0$ if and only if $\dim_H E=-\infty$ and $E$ is the emptyset (see \cite{Falc,Mat} for more details). 

\medskip

We denote by $(\Omega,\mathcal{B},\mathbb{P})$ the probability space on which the random variables considered in this paper are defined. 

\medskip

Finally, if $f$ is a bounded $\mathbb{C}$-valued function over an interval $I$, then $\|f\|_\infty$ stands for $\sup_{t\in I} |f(t)|$.

\section{Proofs of Theorem~\ref{AM}, Theorem~\ref{AM2}, and Corollary~\ref{CW}}\label{mainproof}

The next three sections provide intermediate results yielding Theorem~\ref{AM}. Detailed proofs of these results are given in Section~\ref{proofsIR}. The proof of Theorem~\ref{AM2} is almost the same as that of Theorem~\ref{AM} and we outline it in Section~\ref{AM2proof}. Corollary~\ref{CW} is given in Section~\ref{proofaddf}, and Section~\ref{waekassump} provides weaker assumptions under which these result still hold, or partially hold. 
\medskip

In the next three sections we work under the assumptions of Theorem~\ref{AM}. 

\subsection{Upper bound for the singularity spectra}\label{upperbound}

Let $f$ be a measurable bounded function from $[0,1]$ to $\mathbb{R}$. 

\begin{proposition}\label{Formalism}
Let $m\ge 1$. If $\supp(f^{(m)})\neq\emptyset$ then for every $h\ge 0$ we have $\dim_H  E^{(m)}_f(h)\le (\tau^{(m)}_f)^*(h)$, a negative dimension meaning that $E_f(h)$ is empty. Also,
$$\dim_H \supp(f^{(m)})\le \overline{\dim}_B \supp(f^{(m)})=-\tau^{(m)}_f(0),$$
where $\overline{\dim}_B$ stands for the upper box dimension (see \cite{Falc} for the definition).
\end{proposition}
\begin{remark}  
When $f$ is non-decreasing and $m=1$,  the $L^q$-spectrum $\tau^{(1)}_f$ is nothing but the $L^q$-spectrum of the measure $f'$, and the inequality provided by Proposition~\ref{Formalism} is familiar from the multifractal formalism for measures. Though the proof of the inequality is similar for $m\ge 2$,  for the reader's convenience we will give a proof of Proposition~\ref{Formalism} in Section~\ref{appendix} (see also \cite{Jafw} for similar bounds).
\end{remark}

We first need the following propositions. 
\begin{proposition}\label{nonpol}
With probability 1, $\supp(F')\neq \emptyset$, and the function $F$ is nowhere locally equal to a polynomial over the support of $F'$. Consequently, $\supp(F^{(m)})=\supp(F')$ for all $m\ge 1$.
\end{proposition}

Now for $n\ge 1$, and $(q,t)\in\mathbb{R}^2$ define 
$$
\theta^{(m)}_{F,n}(q,t)=\sum_{w\in \mathscr{A}^n} {\rm Osc}^{(m)}_{F_W}(I_w)^q|I^L_w|^{-t}\text { and } \widetilde \theta^{(m)}_{F,n}(q,t)=\mathbb{E}\big (\theta^{(m)}_{F,n}(q,t)\big ),
$$
with the convention $0^q=0$. Then define
$$
\theta^{(m)}_F(q,t)=\limsup_{n\to\infty} \theta^{(m)}_{F,n}(q,t)\text { and }\widetilde \theta^{(m)}_F(q,t)=\limsup_{n\to\infty}\widetilde \theta^{(m)}_{F,n}(q,t),
$$
as well as 
$$
\tau^{(m)}_{F,b}(q)=\sup\{t\in\mathbb{R}: \theta^{(m)}_F(q,t)=0\}\text { and }\widetilde \tau^{(m)}_{F,b}(q)=\sup\{t\in\mathbb{R}: \widetilde \theta^{(m)}_{F,b}(q,t)=0\}.$$

\begin{proposition}\label{Lqspectrum}
Let $m\ge 1$. With probability 1, for all $q\in\mathbb{R}_+$ we have $\tau^{(m)}_F(q)\geq \tau^{(1)}_{F,b}(q)\ge \widetilde \tau^{(1)}_{F,b}(q)$, and for all $q\le \mathbb{R}^*_-$ we have  $\tau^{(m)}_F(q)\geq \tau^{(m)}_{F,b}(q)\ge \widetilde \tau^{(m)}_{F,b}(q)$. 

Moreover, $\widetilde \tau^{(m)}_{F,b}(q)=\tau(q)$ for all $q< \widetilde q$, where $\widetilde q=\max\{p:\tau(p)=0\}$ (by convention $\max(\emptyset)=\infty$). 
\end{proposition}

\noindent
{\it Proof of the upper bound for the singularity spectra.} Let $m\ge 1$. Recall that $J=\{q\in\mathbb{R}:\tau'(q)q-\tau(q)\ge 0\}$. Since $\tau$ is concave, we have $J\subset (-\infty, \widetilde q]$. Consequently, since $(\tau_F^{(m)})^*$ is concave, due to Proposition~\ref{Lqspectrum}, with probability 1, for all $h\ge 0$ we may have $(\tau_F^{(m)})^*(h)\ge 0$ only if $\tau^*(h)\ge 0$. In this case, we have  $\dim_H E_F^{(m)}(h)\leq (\tau_F^{(m)})^*(h)\le \tau^*(h)$ by Proposition~\ref{Formalism}. Also, since $0$ belongs to $J$, we have $\dim_H \supp(F')\le -\tau(0)$.

\subsection{Lower bound for the singularity spectra}
\medskip

Let $I=\overline{\{\tau'(q):q\in J\}}$. We are going to distinguish the case $h\in \Inte(I)$ and the case $h\in \partial I$.

\subsubsection{\bf The case $h\in \Inte(I)$}\label{lowerbound1}
\medskip

At first we introduce some auxiliary measures. If $q\in \Inte(J)$, $w\in\mathscr{A}^*$ and $n\ge 1$ let
$$
Q_q(w)=\mathbf{1}_{\{Q_W(w)\neq 0\}}|Q_W(w)|^qQ_L(w)^{-\tau(q)}\text{ and }
Y_{q,n}(w)=\sum_{v\in\mathscr{A}^n} Q_q(w\cdot v).
$$

\begin{proposition}\label{muq}$\ $
\begin{enumerate}
\item 
With probability 1, for all $q\in \Inte(J)$ and $w\in\mathscr{A}^*$, the sequence $Y_{q,n}(w)$ converge to a positive limit $Y_q(w)$. Moreover, for every $n\ge 1$, $\sigma(\{Q_U(w):w\in \mathscr{A}^{n-1},U\in\{W,L\}\})$ and $\sigma(\{Y_q(w):w\in \mathscr{A}^n\})$ are independent, and the random variables $Y_q(w)$, $w\in\mathscr{A}^n$, are independent copies of $Y_q(\emptyset)$, that we denote by $Y_q$.  
 
\item 
For every compact subinterval $K$ of $ \Inte(J)$, there exists $p_K>1$ such that
$$\mathbb{E}(\sup\nolimits_{q\in K}Y_q^{p_K})<\infty.$$ 

\item With probability 1, for all $q\in  \Inte(J)$, the function 
\begin{equation}\label{selfsimimeasure}
\mu_q([w])=Q_q(w)Y_q(w),\ w\in\mathscr{A}^* 
\end{equation}
defines a Borel measure on $\mathscr{A}^{\mathbb{N}_+}$. 
\end{enumerate}
\end{proposition}

Recall the definitions given at the end of Section~\ref{intro}. 

For $m\ge 1$, $t\in\mathscr{A}^{\mathbb{N}_+}$, $U\in\{W,L\}$ and $\gamma\in\{-1,0,+1\}$ let 
$$
\underline{\alpha}_U^{(m),\gamma}(t)\ (\text{resp. }\overline{\alpha}_U^{(m),\gamma}(t))=\liminf_{n\to\infty} \ (\text{resp. }\limsup_{n\to\infty})\ -\frac{\log_b {\rm Osc}^{(m)}_{F_U}((t|_{n})^\gamma)}{n}$$
(recall that if $w\in \mathscr{A}^*$, $w^-=w^{-1}$, $w=w^0$ and $w^+=w^{+1}$ are defined in Section~\ref{intro}).
 
\smallskip
 
The next proposition follows directly from the definition of the $m^{\text{\tiny th}}$ oscillation. 
\begin{proposition}\label{exposant}
Let $t\in \mathscr{A}^{\mathbb{N}_+}$ and $\widetilde t=F_L(\pi(t))$. 
\begin{enumerate}
\item Let $r\in (0,1)$ and suppose that 
\begin{equation}\label{inclusion}
\exists\ n_r,\  n'_r\in\mathbb{N},\  I^L_{t|_{n_r}}  \subset B(\widetilde t,r ) \subset I^L_{t|_{n'_r}^-}\cup I^L_{t|_{n'_r}} \cup I^L_{t|_{n'_r}^+}.
\end{equation}
Then
\begin{equation}\label{omo1}
O_W^{(m)}(t|_{n_r})\le  O^{(m)}_F(B(\widetilde t,r)) \le 2^{m-1} O^{(1)}_F(B(\widetilde t,r)) \le  2^{m-1} \sum_{w\in \{t|_{n'_r}^-,t|_{n'_r},t|_{n'_r}^+\}}O^{(1)}_W(w).
\end{equation}

\item Suppose that \eqref{inclusion} holds for all $r>0$ small enough and $\lim_{r\to 0^+} n_r/n'_r=1$. Then,
\begin{equation}\label{exponent}
\frac{\min\{\underline \alpha_W^{(1),\gamma}(t):\gamma=-1,0,+1\}}{\overline\alpha_L^{(1),0}(t)}\le h^{(m)}_F(\widetilde t)\le \frac{\overline  \alpha_W^{(m),0}(t)}{\min \{\underline \alpha_L^{(1),\gamma}(t):\gamma=-1,0,+1\}}.
\end{equation}
\end{enumerate}
\end{proposition}

Recall that for $(q,t)\in\mathbb{R}^2$ we have defined
\begin{equation*}
\Phi(q,t)= \mathbb{E}\Big (\sum_{i=0}^{b-1}\mathbf{1}_{\{W_i\neq 0\}}|W_i|^qL_i^{-t}\Big )
\end{equation*}
and $\tau(q)$ is the unique solution of $\Phi(q,\tau(q))=1$. By construction, we have 
\begin{equation}\label{tau'}
\tau'(q)=-\frac{(\partial\Phi /\partial q)(q,\tau(q))}{(\partial\Phi /\partial t)(q,\tau(q))}=\frac{\mathbb{E}\Big(\sum_{i=0}^{b-1}\mathbf{1}_{\{W_i\neq 0\}}|W_i|^qL_i^{-\tau(q)}\log (|W_i|)\Big)}{\mathbb{E}\Big(\sum_{i=0}^{b-1}\mathbf{1}_{\{W_i\neq 0\}}|W_i|^qL_i^{-\tau(q)}\log (L_i)\Big)} .
\end{equation}
\begin{proposition}\label{keypro}
With probability 1,  for all $q\in \Inte(J)$, for $\mu_q$-almost every $t\in\supp(\mu_q)$,
\begin{enumerate}
\item $\displaystyle
\lim_{n\to\infty} \frac{\log |Q_W(t|_{n})|}{-n}=-\frac{\partial \Phi}{\partial q}(q,\tau(q));\\\ \lim_{n\to\infty} \frac{\log|Q_W((t|_{n})^\gamma)|}{-n}\in\{ -\frac{\partial \Phi}{\partial q}(q,\tau(q)),+\infty\}$, for $\gamma\in\{-1,1\}$;

\smallskip 
\item $\displaystyle
\lim_{n\to\infty} \frac{\log Q_L(t|_{n})}{-n}=\lim_{n\to\infty} \frac{\log Q_L((t|_{n})^\gamma)}{-n}= \frac{\partial \Phi}{\partial t}(q,\tau(q))$, for $\gamma\in\{-1,1\}$;

\smallskip

\item 
$\displaystyle \lim_{n\rightarrow\infty}\frac{\log_b Z_U^{(m)}(t|_{n})}{n}=\lim_{n\rightarrow\infty}\frac{\log Z_U^{(m)}((t|_{n})^\gamma)}{n}=0$, for all $m\ge 1$, $U\in\{W,L\}$ and $\gamma\in\{-1,1\}$.

\smallskip

\item 
$\displaystyle \liminf_{n\rightarrow\infty}\frac{\log Y_q(t|_{n})}{-n}\ge 0$.
\end{enumerate}
\end{proposition}
\smallskip


\noindent
{\it Proof of the lower bound.} Due to \eqref{selfsimimeasure} and  Proposition~\ref{keypro} (1), (2) and (4), with probability 1, for all $q\in\Inte(J)$, we have 
\begin{eqnarray*}
 \displaystyle \liminf_{n\to\infty}\frac{\log(\mu_q([t|_{n}]))}{-n}&\ge& -q\frac{\partial \Phi}{\partial q}(q,\tau(q))-\tau(q) \frac{\partial \Phi}{\partial t}(q,\tau(q))\\
&=&\big (q\tau'(q)-\tau(q)\big )\cdot \frac{\partial \Phi}{\partial t}(q,\tau(q))>0,\ \mu_q\text{-}a.e.
\end{eqnarray*}
($\frac{\partial \Phi}{\partial t}(q,\tau(q))>0$ due to our choice $L_i\in (0,1)$). Consequently, $\mu_q$ is atomless, and defining $\nu_q=\mu_q\circ \pi^{-1}\circ F_L^{-1}$, we have $\nu_q(I^L_w)=\mu_q([w])$ for all $w\in\mathscr{A}^*$. Thus, $$\liminf_{n \to\infty}\frac{\log \nu_q(I^L_n(t))}{-n}\ge q\frac{\partial \Phi}{\partial q}(q,\tau(q))+\tau(q) \frac{\partial \Phi}{\partial t}(q,\tau(q)), \ \nu_q \text{-almost everywhere},$$
where $I^L_n(t)$ is the unique interval $I^L_w$ of generation $n$ containing $t$.

Now, Proposition~\ref{keypro} (2) and (3) as well as \eqref{self-sim2'} also yield
$$\lim_{n\to\infty}\frac{\log|I_n^L(t)|}{-n}=\frac{\partial \Phi}{\partial t}(q,\tau(q))>0,\ \nu_q\text{-almost everywhere},$$
hence 
$$
\liminf_{n\to\infty}\frac{\log \nu_q(I^L_n(t))}{\log|I_n^L(t)|}\ge  q\tau'(q)-\tau(q),\ \nu_q\text{-almost everywhere}.
$$
Consequently, we can apply the mass distribution principle (\cite{Pe}, Lemma 4.3.2) and we obtain $\dim_H (\nu_q)\ge q\tau'(q)-\tau'(q)=\tau^*(\tau'(q))$. 
\smallskip

We can also deduce from Proposition~\ref{keypro} that for $\mu_q$-almost every $t$, for all $m\ge 1$,
\[
\begin{cases}
\min\{\underline \alpha_W^{(1),\gamma}(t):\gamma=-1,0,+1\}=\overline  \alpha_W^{(m)}(t)= -\frac{\partial \Phi}{\partial q}(q,\tau(q))/\log(b),\\
\min \{\underline \alpha_L^{(1),\gamma}(t):\gamma=-1,0,+1\}=\overline\alpha_L^{(1)}(t)=\frac{\partial \Phi}{\partial t}(q,\tau(q))/\log(b).
\end{cases}
\] 

These properties imply that at $\nu_q$-amost every $\widetilde t$, for $r\in (0,1)$ small enough, we can find integers $n_r$ and $n'_r$ such that \eqref{omo1} holds with $\lim_{r\to 0^+}n_r/n'_r=1$,  and we have for all $m\ge 1$ $h^{(m)}_F(t)=\tau'(q)$. Due to Proposition~\ref{hinfty}, we also have $h^{(\infty)}_F(t)=\tau'(q)$. Since $\dim_H (\nu_q)\ge \tau^*(\tau'(q))$ we have the desired lower bound for the dimensions of the sets $E^{(m)}_F(\tau'(q))$, $m\in\mathbb{N}\cup\{\infty\}$. The case $q=0$ yields $\dim_H \supp(F')\ge \dim_H E^{(1)}_F(\tau'(0))\ge -\tau(0)$.
\smallskip

Combining this with Proposition~\ref{Lqspectrum} we obtain that, with probability 1,  for all $m\in\mathbb{N}_+$, we have $(\tau^{(m)})_F^*=\tau^*$ over $\Inte(I)$. Since we also have $\tau^{(m)}_F\ge \tau^{(m)}_{F,b}\ge \tau$ over $J$, this yields $\tau^{(m)}_F=\tau^{(m)}_{F,b}=\tau$ over $J$.

\subsubsection{\bf The case $h\in \partial I$}\label{lowerbound2}
\medskip

Recall that $\underline q=\inf J$ and $\overline q =\sup J$. Let 
$$\underline h=\lim\limits_{q\to \overline q} \tau'(q),\bar h =\lim\limits_{q\to \underline q} \tau'(q), \bar d=\lim\limits_{q\to \overline q} \tau'(q)q-\tau(q), \underline d=\lim\limits_{q\to \underline q} \tau'(q)q-\tau(q).$$
Then $\partial I=\{\underline h,\bar h\}$, $\bar d=\tau^*(\underline h)$ and $\underline d=\tau^*(\bar h)$. Moreover, with probability 1, $\bar d=\tau^*(\underline h)=(\tau^{(m)}_F)^*(\underline h)$ and $\underline d=\tau^*(\bar h)=(\tau^{(m)}_F)^*(\bar h)$ for any $m\ge 1$. 

The difficulty in the study of $E_F^{(m)}(h)$ when $h\in \{\underline h,\bar h\}$ comes from the fact that there is no simple choice of a measure carried by $E_F^{(m)}(h)$ and whose Hausdorff dimension is larger than or equal to $(\tau^{(m)}_F)^*(h)$. Even, it is not obvious to construct a point belonging to $E_F^{(m)}(h)$. Neverthless such a measure can be constructed. 

\medskip

\noindent
{\bf A measure $\mu_q$ partly carried by $E_F^{(m)}(h)$, for $(q,h)\in\{(\underline q,\underline h),(\overline q,\overline h)\}$.}

\medskip

\noindent
{\bf 1. The case $q\not\in\{-\infty,\infty\}$.}

\medskip

Let $W_q(w)=\big (\mathbf{1}_{\{W_i(w)\neq 0\}}|W_i(w)|^q(w)L_i(w)^{-\tau(q)}\big )_{0\le i\le b-1}$. We have 
$$\tau^*(\tau'(q))=\varphi_{W_q}'(1)=-\mathbb{E}(\sum_{i=0}^{b-1}W_{q,i}\log_b W_{q,i})=0.$$
Moreover, $\varphi_{W_q}(p)>-\infty$ in a neighborhood of $1^+$. Consequently, it follows from Theorem 2.5 of \cite{Liu} that, with probability 1, for all $w\in\mathscr{A}^*$, the martingale 
$$
Y_{q,n}(w)=-\sum_{u\in \mathscr{A}^n}Q_q(w\cdot u)\log Q_q(w\cdot u), \text{ with } Q_q(w)=\prod_{k=1}^nW_{q,w_k}(w|_{k-1}),
$$
converges to a limit  $Y_q(w)$ ($Y_q(\emptyset)=Y_q$) as $n\to\infty$. Moreover, by construction, the branching property $
Y_q(w)=\sum_{i=0}^{b-1}W_{q,i}(w) Y_q(wi)
$ holds, the random variables $Y_q(w)$, $w\in\mathscr{A}^*$, are identically distributed,  and for $\gamma>0$ we have $\mathbb{E}(Y_q^\gamma)<\infty$ if and only if $\gamma<1$. 

We deduce from the branching property and our assumption on the probability that the components of $W$ vanish that the event $\{Y_q=0\}$ is measurable with respect to the tail $\sigma$-algebra $\bigcap_{N\ge 1}\sigma(W(w):w\in \bigcup_{n\ge N} \mathscr{A}^n)$. Consequently, $\mathbb{P}(Y_q>0)=1$ since $\mathbb{E}(Y_q)>0$, and with probability 1, the branching property makes it possible to define on $\mathscr{A}^{\mathbb{N}_+}$ a measure $\mu_q$ by the formula
\begin{equation}\label{muh}
\mu_q([w])=Q_q(w) Y_q(w).
\end{equation}
\begin{proposition}\label{endpoints}
Let $h\in\{\bar h, \underline h\}$ and $q\not\in\{-\infty,\infty\}$ such that $h=\tau'(q)$. With probability 1, there exists a Borel set $E_{h}\subset \mathscr{A}^{\mathbb{N}_+}$ of positive $\mu_{q}$-measure such that for all $t\in E_{h}$ the same conclusions as in Proposition~\ref{keypro} (1) (2) (3) hold. 
\end{proposition}
Then, the same arguments as in Section~\ref{lowerbound1} yield $\nu_q(E_F^{(m)}(h))>0$, hence $E_F^{(m)}(h)$ is not empty and we get the desired lower bound $\dim_H E_F^{(m)}(h)\ge 0$ since $\tau^*(h)=0$. 

\medskip

\noindent
{\bf 2. The case $q\in\{-\infty,\infty\}$.}

\medskip

Let $(q_k)_{k\geq 0}$ be an increasing (resp. decreasing) sequence converging to $q$ if $q=\infty$ (resp. $q=-\infty$).  For every $k\ge 0$ and $w\in \mathscr{A}^k$ let
$$W_{q_k}(w)=\big(\textbf{1}_{\{W_i(w)\neq 0\}}\cdot|W_i(w)|^{q_k}L_i(w)^{-\tau(q_k)} \big)_{0\leq i \leq b-1}.$$
Then, for $n\ge 1$ and $w\in \mathscr{A}^*$~let 
$$Y_{q,n}(w)=\sum_{u\in \mathscr{A}^n} Q_q(w\cdot u), \text{ with } Q_q(w)=\prod_{k=}^{|w|}W_{q_{k-1},w_k}(w|_{k-1}),$$
and simply denote $Y_{q,n}(\emptyset)$ by $Y_{q,n}$. The sequence $(Y_{q,n}(w))_{n\ge 1}$ is a non-negative martingale of expectation 1 which converges almost surely to a limit that we denote by $Y_q(w)$ ($Y_q$ if $w=\emptyset$). Since the set $\mathscr{A}^*$ is countable, all these random variable are defined simultaneously. Moreover, the branching property $
Y_q(w)=\sum_{k=0}^{b-1}Q_{q,i}(w)Y_q(w i)$ also holds. Notice that by construction, given $k\ge 1$, the random variables $Y_q(w)$, $w\in\mathscr{A}^k$, are independent and identically distributed. 

\begin{proposition}\label{qk}
The sequence $(q_k)_{k\ge 0}$ can be chosen so that there exists $a>0$ such that for all $w\in \mathscr{A}^*$ the sequence $(Y_{q,n}(w))_{n\ge 1}$ converges in $L^2$ norm to a limit $Y_q$ and $\|Y_q(w)\|_2=O\big (b^{a|w|/\log (|w|)}\big )$. 
\end{proposition}
Fix a sequence $(q_k)_{k\ge 0}$ as in the previous proposition. For the same reason as in the case $q\not\in \{-\infty,\infty\}$, we have $\mathbb{P}(Y_q>0)=1$ and with probability 1, the branching property makes it possible to define on $\mathscr{A}^{\mathbb{N}_+}$ a measure $\mu_q$ by the formula \eqref{muh}. 

\begin{proposition}\label{endpoints2}
Let $h\in\{\bar h, \underline h\}$ and $q\in\{-\infty,\infty\}$ such that $h=\lim_{J\ni q'\to q}\tau'(q)$. Let $\nu_q=\mu_q\circ\pi^{-1}\circ F_L^{-1}$. With probability 1, for every $m\ge 1$,  we have $h^{(m)}_F({t})=h$ $\nu_{q}$-almost everywhere and  $\dim_H (\nu_q)\ge \tau^*(h)$.
\end{proposition}

\begin{remark}
In the case $q\not\in\{-\infty,\infty\}$, it is possible to construct $\mu_q$ as in the case $q\in\{-\infty,\infty\}$ by using a sequence $(J\ni q_k)_{k\ge 0}$ converging to $q$. This avoids to require to Theorem 2.5 of \cite{Liu} which is a strong result. Nevertheless, we are able to use this alternative only if $\varphi_W(q)>-\infty$ for some $q<-1$. This is the case under the assumptions of Theorem~\ref{AM}, but this does not always hold under the weaker assumptions provided by Section~\ref{waekassump}. 
\end{remark}

\subsection{The $L^q$-spectra of $F$}\label{Lqspec}

We have seen at the end of Section ~\ref{lowerbound1} that, with probability 1, for all $m\in\mathbb{N}$, $\tau^{(m)}_F(q)=\tau^{(m)}_{F,b}(q)=\tau(q)$ over $J=[\underline q, \overline q]$. It remains to show that $\tau^{(m)}_F$ is differentiable at $\overline q$ (resp. $\underline q$) and linear over $[\overline q, \infty)$ (resp. $(-\infty, \underline q]$) if $\overline q<\infty$ (resp. $\underline q>-\infty$). We treat the case $\overline q<\infty$ and leave the case $\underline q>-\infty$ to the reader. 

At first we notice that the equality $\tau^{(m)}_F=\tau^{(m)}_{F,b}=\tau$ over $J$ implies that $(\tau^{(m)}_F)'(\overline q^-)=\tau^{(m)}_{F,b}(\overline q)/\overline q=\tau(\overline q)/\overline q=\underline h$. Also, by concavity of $\tau^{(m)}_F$, we have $\tau^{(m)}_F(q)\le \tau^{(m)}_F(\overline q)+(\tau^{(m)}_F)'(\overline q^-)(q-\overline q)=  \tau(\overline q)+\tau'(\overline q)(q-\overline q)=\underline hq$. To get the other inequality, and so the differentiability of $\tau^{(m)}_F$ at $\overline q$, we use a simple idea inspired by the work achieved in \cite{Mol} which focuses on $\tau^{(1)}_{F,b}$ in the case when the components of $W$ are non-negative and $L=(1/b,\dots,1/b)$. If $q\ge \overline q$ and $t\in\mathbb{R}$, we have 
$$
\sum_{w\in\mathscr{A}^n} {\rm Osc}^{(m)}_{F_W}(I_w)^q|I^L_w|^{-t}\le \big[ \sum_{w\in\mathscr{A}^n}  {\rm Osc}^{(m)}_{F_W}(I_w)^{\overline q}|I^L_w|^{-\overline q t/q}\big ]^{q/\overline q},
$$
 because $q/\overline q\ge 1$. Consequently, by definition we have $\tau^{(m)}_{F,b}(q)\ge (q/\overline q)\cdot \tau^{(m)}_{F,b}(\overline q)=q\underline h$.  This, together with Proposition~\ref{Lqspectrum}, yields $\tau^{(m)}_F(q)\ge \tau^{(m)}_{F,b}(q)\ge \underline h q$ for $q\ge \overline q$.

\smallskip

It remains to discuss the strict concavity of $\tau$ over $J$. Suppose $\tau$ is affine over a non trivial sub-interval $J'$ of $J$. The analyticity of $\tau$ implies that it is affine over $J$ (in fact over $\mathbb{R}$ under our assumptions), which is equivalent to saying that for all $q,q'\in J$ and $\lambda\in [0,1]$ we have 
\begin{equation}\label{Phi'}
\Phi\big (\lambda q+(1-\lambda)q', \lambda \tau(q)+(1-\lambda)\tau(q')\big )=0,
\end{equation}
where $\Phi$ is defined in \eqref{Phi}. Let $\lambda\in (0,1)$ and $q\neq q'\in J$. Applying the H\"older inequality to $\sum_{i=0}^{b-1}\mathbf{1}_{\{W_i\neq 0\}}|W_i|^{\lambda q}L_i^{-\lambda \tau(q)}|W_i|^{(1-\lambda)q'}L_i^{-(1-\lambda)\tau(q')}$ shows that, in order to have \eqref{Phi'} it is necessary and sufficient that there exists $C$ such that
$$\mathbf{1}_{\{W_i\neq 0\}}|W_i|^{q}L_i^{-\tau(q)}=C\mathbf{1}_{\{W_i\neq 0\}}|W_i|^{q'}L_i^{-\tau(q')}$$
almost surely. Thus, there exists $H>0$, the slope of $\tau$, such that $|W_i|=L_i^H$ for all $i$, conditionally on $W_i\neq 0$. If the components of $W$ are non-negative almost surely, by construction this implies $\mathbb{E}(\sum_{i=0}^{b-1}L_i^H)=1$, hence $H=1$ and $W=L$, the situation we have discarded. Otherwise, we have $\mathbb{E}(\sum_{i=0}^{b-1}L_i^H)>1$ hence $H\in (0,1)$.

\subsection{Proof of Theorem \ref{AM2}}\label{AM2proof}
We only have to deal with the exponent $h=\tau'(0)=\infty$. The rest of the study is similar to that achieved in the previous sections. 

For $w\in\mathscr{A}^*$ let $\widetilde W(w)=\big (\mathbf{1}_{\{W_i(w)\neq 0\}}L_i(w)^{-\tau(0)}\big )_{0\le i\le b-1}$. By construction the components of $\widetilde W$ are non negative, we have $\varphi_{\widetilde W}(1)=0$, and $\varphi_{\widetilde W}'(1)>0$. Consequently, the Mandelbrot measure on $\mathscr{A}^{\mathbb{N}_+}$ defined as $\mu_0=F_{\widetilde W}'$ (with the notations of Theorems A and B) is positive with probability 1. Moreover, it follows from  the study achieved in \cite{B2} that $\dim_H \nu_0=-\tau(0)$, where $\nu_0=\mu_0\circ \pi^{-1}\circ F_L^{-1}$.

Now, for $a\in (0,1)$, we define $W^{(a)}=(|W_i|\wedge a)_{0\le i\le b-1}$.  We have $h=\lim_{a\to0} h(a)$, where $h(a)=-\mathbb{E}\big (\sum_{i=0}^{b-1}\mathbf{1}_{\{W_i\neq 0\}}L_i\log (W^{(a)}_i)\big )$. By using the same techniques as in Section \ref{proofsIR} we can prove that, with probability 1, for $\mu_0$-almost every $t$, we have $$
\displaystyle
\lim_{n\to\infty} \frac{\log|Q_{W^{(a)}}((t|_{n})^\gamma)|}{-n}=h(a), \ \gamma\in\{-1,0,1\}.
$$
Also, due to our assumptions and Proposition~\ref{xmoments}, we have 
$
\displaystyle \lim_{n\rightarrow\infty}\frac{\log Z_W^{(1)}((t|_{n})^\gamma)}{n}=0,
$
for all $\gamma\in\{-1,0,1\}$.
Consequently, for $\mu_0$-almost every $t$, $\min(\underline\alpha_W^{(1),\gamma}(t):\gamma\in\{-1,0,1\})\ge h(a)$. Since this holds for every $a\in(0,1)$, letting $a$ tend to 0 yields $\min(\underline\alpha_W^{(1),\gamma}(t):\gamma\in\{-1,0,1\})=\infty$ for $\mu_0$-almost every $t$. Since there exists $\overline a>0$ such that $\overline\alpha_L^{(1),0}(t)\le \overline a$ for all $t$ (see Lemma~\ref{controlI}) we conclude thanks to Proposition~\ref{exposant} that for $\nu_0$-almost every $t$ we have $h_F^{(1)}(t)=\infty$. 

\subsection{Proof of Corollary \ref{CW}}\label{proofaddf}

Fix $1\le m\in \mathbb{N}$. Recall that $q_m$ is the unique real number such that $\tau(q_m)=q_mm-1$.

Let $C>0$ such that $\mathrm{Osc}_{f}^{(m)}(B)\le C|B|^m$ for all subintervals $B$ of $[0,F_L(1)]$. 

For $r>0$ let $\mathcal{B}_r$ be a family of disjoint closed intervals $B$ of $[0,F_L(1)]$ of  radius $r$ with centers in $\supp(F^{(m)})$. For any $q\in\mathbb{R}_+$ we have
\begin{eqnarray*}
\sum_{B\in\mathcal{B}_r} \mathrm{Osc}_{F+f}^{(m)}(B)^q\cdot r^{-t} &\le& 2^q\sum_{B\in\mathcal{B}_r} \big (\mathrm{Osc}_{F}^{(m)}(B)^q+\mathrm{Osc}_{f}^{(m)}(B)^q\big )r^{-t} \\
&\le& (2C)^q\cdot \Big(\sum_{B\in\mathcal{B}_r} \mathrm{Osc}_{F}^{(m)}(B)^qr^{-t} +\sum_{B\in\mathcal{B}_r}r^{qm-t}\Big).
\end{eqnarray*}
By the definition of $\tau_G^{(m)}(q)$ this yields $\tau_{G}^{(m)}(q)\ge \min (\tau_F^{(m)}(q),qm-1)$ so $(\tau_{G}^{(m)})^*(h)\le \tau^*(h)$ for $h\in [0, \tau'(q_m)]$ (we have used the equality $\tau_F^{(m)}=\tau$) and $(\tau_{G}^{(m)})^*(h)=1$ for $h>m$. 

On the other hand, since we assumed that $f^{(m)}$ does not vanish, we deduce from Theorem~\ref{AM2} that for any $t\in [0,F_L(1)]$ we have $h^{(m)}_G(t)=h^{(m)}_F(t)$ if $h^{(m)}_F(t)< m$ and $h^{(m)}_G(t)=m$ if $h^{(m)}_F(t)>m$. Thus 
\begin{equation*}
\big(\tau_{G}^{(m)}\big)^*(h)\ge \dim_H E_G^{(m)}(h)=\left\{
\begin{array}{ll}
\tau^*(h), & \text{if } h\in[0,m);\\
1, & \text{if } h=m.
\end{array}
\right.
\end{equation*}
This implies that $(\tau_{G}^{(m)})^*$ is equal to $\tau^*$ over $[0, \tau'(q_m)]$ and equal to $h\mapsto \tau^*( \tau'(q_m))+q_m(h-\tau'(q_m))$ over $[\tau'(q_m),m]$. Taking the inverse Legendre transform implies that $\tau_{G}^{(m)}(q)= \min (\tau_F^{(m)}(q),qm-1)$ for all $q\ge 0$. 

\subsection{Weaker assumptions}\label{waekassump}

\subsubsection*{Theorem \ref{AM}} If we only assume that $\varphi_W>-\infty$ in a neighborhood $\widetilde{J}$ of $[0,1]$, then the multifractal formalisms holds for $F$ at each $h=\tau'(q)$ for all $q\in \widetilde{J}\cap J$. Also, the functions $\tau^{(m)}_F$ and $\tau$ coincide over $\widetilde{J}\cap J$. If, moreover, there exists $q_0\in\widetilde{J}$ such that $\tau^*(\tau'(q_0))=0$ then either $q_0>0$ and $\tau^{(m)}_F(q)=\tau'(q_0)q/q_0$ over $[q_0,\infty)$ or $q_0<0$ and $\tau^{(m)}_F(q)=\tau'(q_0)q/q_0$ over $(-\infty,q_0]$.

\medskip

\subsubsection*{Theorem \ref{AM2}} The same discussion as for Theorem \ref{AM} holds, except that $\widetilde{J}$ is a neighbor of $[0,1]$ in $\mathbb{R}_+$.

\section{Proofs of the intermediate results of Section \ref{mainproof}}\label{proofsIR}

\subsection{Proofs of the results of Section \ref{upperbound}}$\ $\\

\noindent
{\it Proof of Proposition~\ref{Formalism}.} This is a consequence of Proposition~\ref{Appendix}.

\medskip

\noindent
{\it Proof of Proposition~\ref{nonpol}.} The result could be obtained after achieving the multifractal analysis using the first order oscillation exponent. Nevertheless we find valuable to have a proof only based on the the functional equation satisfied by the process~$F$. 

We assumed that $\mathbb{P}(\sum_{i=0}^{b-1}\textbf{1}_{\{W_i\neq 0\}}\geq 2)=1$. Consequently, it follows from the definition of $F_W$ that the event $\{Z_W^{(1)}=0\}$ is measurable with respect to the tail $\sigma$-algebra $\bigcap_{n\ge 0}\sigma(\{W(w):w\in \cup_{p\ge n}\mathscr{A}^p\})$ which contains only sets of probability 0 or 1. Since $\mathbb{E}(F_W(1))=1$, we have $Z_W^{(1)}>0$ with positive probability, hence almost surely. So $\supp(F')\neq\emptyset$ almost surely. 

Now we prove that $F$ is nowhere locally equal to a polynomial function over the support $\supp(F')$.

At first, suppose that there exists $0\le i\le b-1$ such that $\mathbb{P}(W_i=0)>0$. Then, with probability 1, the interior of $\supp(F')$ is empty, since for every $w\in\mathscr{A}^*$ the probability that there exists $v\in\mathscr{A}^*$ such that $W_i(w\cdot v)=0$ is equal to 1. Thus $F$ is nowhere locally equal to a polynomial function over $\supp(F')$. 

Now suppose that the components of $W$ do not vanish and that there is a positive probability that there exists an interval $I^L_w$ over which $F$ is equal to a polynomial. Equivalently, $F^{[w]}=F_W^{[w]}\circ (F_L^{[w]})^{-1}$ is a polynomial function. Due to the statistical self-similarity of the construction, the probability that $F$ be itself a polynomial function is positive. Moreover, $F$ is almost surely the uniform limit of the sequence $(F_n=F_{W,n}\circ F_{L,n}^{-1})_{n\ge 1}$. The functions $F_n$ are piecewise linear, and because we assumed $W\neq L$ and the vectors $(W(w),L(w))$, $w\in\mathscr{A}^*$, are independent, with probability 1, for every $w\in\mathscr{A}^*$, there are infinitely many $n$ such that the restriction of $F_n$ to $I^L_w$ is not linear, thus non differentiable. Consequently, the event $\{F \text{ is a polynomial}\}$ is measurable with respect to the tail $\sigma$-algebra $\bigcap_{n\ge 0}\sigma(\{W(w),L(w):w\in \cup_{p\ge n}\mathscr{A}^p\})$, so it has a probability equal to 1. For $0\le i\le b-2$, let $x_i=F_L (i/b)$. By construction, we have $(W_i/L_i)(F^{[i]})'\big (F^{[i]}_L(1)\big )=F'(x_{i+1}^-)=F'(x_{i+1}^+)=(W_{i+1}/L_{i+1})(F^{[i+1]})'(0)$.  Due to the independence between $(W,L)$, $F^{[i]}$ and $F^{[i+1]}$, we see that all the terms in the previous equality must be deterministic, except if ${F^{[i]}}'\big (F^{[i]}_L(1)\big )=(F^{[i+1]})'(0)=0$ almost surely. In this later case, by statistical self-similarity we also have $F'\big (F_L(1)\big )\equiv F'(0)\equiv 0$, and by induction over $n\ge 0$ we see that $F'$ vanishes at all the endpoints of the intervals $I^L_w$, $w\in\mathscr{A}^n$. Thus $F'\equiv 0$ and $F$ is constant.  This is in contradiction with $F(0)=0$ and $\mathbb{E}\big (F(F_L(1))\big )=\mathbb{E}\big (F_W(1))=1$. Consequently, $(W,L)$ must be deterministic. Since we supposed that $W\neq  L$, the assumption $\sum_{i=0}^{b-1}W_i=1=\sum_{i=0}^{b-1}L_i$ implies that $|W_i|>L_i$ for some $0\le i\le b-1$. Let us write $|W_i|=L_i^{H}$ with $H<1$ (recall that $L_i<1$). Then, denoting by $i^{\cdot n}$ the word consisting in $n$ letters $i$, we have  ${\rm Osc}^{(1)}(F,I^L_{i^{\cdot n}})\ge |W_i|^n=|I^L_{i^{\cdot n}}|^H$ so $F$ is not $C^1$. This is a new contradiction, hence $F$ is nowhere locally equal to a polynomial function.
 
\medskip

\noindent
{\it{Proof of Proposition $\ref{Lqspectrum}$}.} We first establish the inequalities $\tau^{(m)}_F\ge \tau^{(1)}_{F,b}$ over $\mathbb{R}_+$ and $\tau^{(m)}_F\ge \tau^{(m)}_{F,b}$ over $\mathbb{R}^*_-$. By applying Theorem~2.3 in~\cite{BJMpartII} to $L$ we immediately have the following lemma:

\begin{lemma}\label{controlI}
There exist $\underline{a},\overline{a}>0$ such that, with probability 1, there exists $n_0\in\mathbb{N}$ such that for $n\ge n_0$, $ b^{-n \overline{a}}\le \inf_{w\in\mathscr{A}^n}|I^L_w|\le \sup_{w\in\mathscr{A}^n}|I^L_w|\le b^{-n \underline{a}}$. Moreover, with probability 1, for every $\varepsilon>0$, there exists $n_\varepsilon$ such that
\begin{equation}\label{ratio1}
\forall\, n\ge n_\varepsilon, \ b^{-n\varepsilon} \le \inf_{w\in\mathscr{A}^n} \inf_{0\le i\le b-1}\frac{|I^L_{wi}|}{|I^L_{w}|}\le 1.
\end{equation}
\end{lemma}

For $\mathbb{P}$-almost every $\omega\in\Omega$, we fix $\varepsilon>0$, $n_0$ and $n_\varepsilon$ as in Lemma~\ref{controlI}.

Let $n'_\varepsilon=\max (n_0,n_\varepsilon)$. Fix $0<r\le \min_{w\in \mathscr{A}^{\mathbb{N}_+}_{n'_\varepsilon+1}} |I^L_w|$.

Let $\mathcal{B}_r$ be a family of disjoint closed intervals $B$ of radius $r$ with centers in $\supp(F')$. If  $B \in \mathcal{B}_r$, by construction we can find three disjoint intervals $I^L_{w_k}$, $k=1,2,3$, with $|w_k|\ge n'_\varepsilon+1$ such that $B\subset I^L_{w_1}\cup I^L_{w_2}\cup I^L_{w_3}$ and $r\le |I^L_{w_k}|\le r b^{|w_k|\varepsilon}$. Also, $|I^L_{w_k}|\le b^{-|w_k| \underline{a}}$ so $ b^{|w_k|\varepsilon}\le r^{-\varepsilon/\underline{a}}$. Thus  $r\le |I^L_{w_k}|\le r ^{1-\varepsilon/\underline{a}}$. 

We have $ O^{(m)}_F( B)\le 2^{m-1}O^{(1)}_F( B)\le 2^{m-1} \sum_{k=1}^3 O^{(1)}_F(I^L_{w_k})$, so for $q\ge 0$ and $t\in\mathbb{R}$ we have 
$$
O^{(m)}_F( B)^q r^{-t} f(t,r) \leq  2^{(m-1)q}3^q \cdot \sum_{k=1}^3 O^{(1)}_F(I^L_{w_k})^q |I^L_{w_k}|^{-t},
$$
with $f(t,r)=1$ if $t<0$ and $f(t,r)=r ^{t\varepsilon/\underline{a}}$ otherwise. Moreover, each so selected interval $I^L_{w_k}$ meets at most  $1+r ^{-\varepsilon/\underline{a}}$ elements of $\mathcal{B}_r$. Consequently,
\begin{equation}\label{comptau}
\sum_{B\in\mathcal{B}_r} O^{(m)}_F( B)^q r^{-t} f(t,r)\le 2^{(m-1)q}3^q (1+r ^{-\varepsilon/\underline{a}}) \sum_{n\ge  n'_\varepsilon+1}\theta^{(1)}_{F,b,n}(q,t).
\end{equation}
Suppose that $\tau^{(1)}_{F,b}(q)>-\infty$; otherwise there is nothing to prove. Due to the existence of $\underline{a}$, by definition of $\tau^{(1)}_F(q)$, if $t<\tau^{(1)}_{F,b}(q)$ then we have   $\sum_{n\ge  n'_\varepsilon+1}\theta^{(1)}_{F,n}(q,t)<\infty$. Then, it follows from \eqref{comptau} and the definition of $\tau^{(m)}_F(q)$ that $ \tau^{(m)}_F(q)\ge t-(1+|t|)\varepsilon/\underline{a}$. Since $\varepsilon$ is arbitrary, we get $\tau^{(m)}_F(q)\ge \tau^{(1)}_{F,b}(q)$. 

\smallskip

On the other hand, for each $B\in \mathcal{B}_r$ there exists $I^L_w$ of maximal length included in $B$. We have $2r b^{-|w|\varepsilon}\le |I^L_w|\le 2r $. This yields $2r\le b^{|w|(\varepsilon-\underline{a})}$ so $(2r)^{\varepsilon/(\underline{a}-\varepsilon)}\le b^{-|w|\varepsilon}$ whenever $\underline{a}>\varepsilon$. consequently, for $\varepsilon$ small enough, we have $(2r)^{1+ \varepsilon/(\underline{a}-\varepsilon)}\le |I^L_w|\le 2r $. Thus, if $q<0$ we have 
$$
O^{(m)}_F( B)^q (2r)^{-t} f(t,r)\le O^{(m)}_F( I^L_w)^q |I^L_{w}|^{-t},
$$
where $f(t,r)=1$ if $t\ge 0$ and $f(t,r)= (2r)^{-t\varepsilon/(\underline{a}-\varepsilon)}$ otherwise.
Since the elements of $\mathcal{B}_r$ are pairwise disjoint, this implies
\begin{equation}\label{comptau'} \sum_{B\in \mathcal{B}_r} O^{(m)}_F(B)^q (2r)^{-t} f(t,r)\leq \sum_{n\ge  n'_\varepsilon+1}\theta^{(m)}_{F,n}(q,t)
\end{equation}
and the same arguments as when $q\ge 0$  yield $\tau^{(m)}_F(q)\ge \tau^{(m)}_{F,b}(q)$.
\smallskip

To see that, with probability 1, $\tau^{(m)}_{F,b}\ge \widetilde \tau^{(m)}_{F,b}$, due to the concavity of  $\tau^{(m)}_{F,b}$ and $\tau^{(m)}_{F,b}$, it is enough to show that given $q\in\mathbb{R}$, we have $\tau^{(m)}_{F,b}(q)\ge \widetilde \tau^{(m)}_{F,b}(q)$. 

Let $(q,t)\in\mathbb{R}^2$, and suppose that $q< \max\{p:\varphi_W(p)=0\}$. Due to Proposition~\ref{xmoments} we have $\psi(q,t)<\infty$. By using (\ref{self-sim2}) we get $\widetilde \theta^{(m)}_{F,n}(q,t)=\Phi(q,t)^n\Psi(q,t)$ for al $n\ge 1$. This yields $\widetilde\tau^{(m)}_{F,b}(q)=\tau(q)$. Also, if $t<\tau(q)$
then $\Phi(q,t)<1$ and $\sum_{n\ge 1} \widetilde \theta^{(m)}_{F,n}(q,t)<\infty$ so $\sum_{n\ge 1}  \theta^{(m)}_{F,n}(q,t)<\infty$ almost surely. This yields $t<\tau^{(m)}_{F,b}(q)$. Since $t$ is arbitrary we get $\tau^{(m)}_{F,b}(q)\ge \widetilde\tau^{(m)}_{F,b}(q)$. 

To finish the proof, we notice that by construction, we have $\tau(p)=0$ if and only if $\varphi_W(p)=0$. 
\subsection{Proofs of the results of Section \ref{lowerbound1}}$\ $\\

\noindent
{\it Proof of Proposition~\ref{muq}.}  
This proof could be deduced from those of Lemma 4 and Corollary 5 of \cite{B2}. For reader's convenience, we provide it.  

\medskip

\noindent
$\bullet$ {\it Proof of (1) and (2).}  For $q\in  \Inte(J)$ and $w\in\mathscr{A}^*$ let 
$$
W_q(w)=\big (\mathbf{1}_{\{W_0\neq 0\}}|W_0|^qL_0^{-\tau(q)},\dots, \mathbf{1}_{\{W_{b-1}\neq 0\}}|W_{b-1}|^qL_{b-1}^{-\tau(q)}\big )(w).
$$ 
The function $ \Phi$ can be extended to an analytic function in a complex neighborhood of $J\times \mathbb{C}$ by 
$$
 \Phi(z,t)=\mathbb{E}\Big (\sum_{i=0}^{b-1} \mathbf{1}_{\{W_i\neq 0\}}|W_i|^zL_i^{-t}\Big).
$$
For each $q\in \Inte(J)$ we have $\frac{\partial  \Phi (q,\tau(q))}{\partial t}=- \mathbb{E}\Big (\sum_{i=0}^{b-1} W_{q,i}\log (L_i)\Big )>0$ and $\Phi(q,\tau(q))=1$, so there exists a neighborhood $V_q$ of $q$ in $\mathbb{C}$ such that for each $z\in V_q$ there exists a unique $\tau(z)$ such that $ \Phi(z,\tau(z))=1$. Moreover, the mapping $z\mapsto \tau(z)$ is analytic. We define 
$$
W_z(w)=\big (\mathbf{1}_{\{W_0\neq 0\}}|W_0|^zL_0^{-\tau(z)},\dots, \mathbf{1}_{\{W_{b-1}\neq 0\}}|W_{b-1}|^zL_{b-1}^{-\tau(z)}\big )
$$
as well as the mapping 
$$
(z,p)\in V_q\times [1,\infty)\mapsto M(z,p)= \sum_{i=0}^{b-1}\mathbb{E}(|W_{z,i}|^p).
$$ 

The property $\tau^*(\tau'(q))>0$ is equivalent to $\displaystyle \frac{\partial M}{\partial p}(q,1^+)<0$ , so there exists $p_q>1$ and a open neighborhood $ V'_q$ of $q$ in $J$ such that $\sup_{q'\in V'_q}M(q',p)<1$ for all $p\in (1,p_q]$ (because $p\mapsto  M(q',p)$ is convex and $M(q',1)=1$). Now, we fix $K$ a non-trivial compact subinterval of $\Inte(J)$. It is covered by a finite number of such $V'_{q_i}$ so that if $V'_K= \bigcup_i V'_{q_i}$ we have $\sup_{q\in V'_K} M(q,p_K)<1$, where $p_K=\inf_i p_{q_i}$.  By a comparable procedure we can now find a complex neighborhood $V_K$ of $V_K$ such $\sup_{z\in V_K} M(z,p_K)<1$.

To prove the almost sure simultaneous convergence of the martingales $(Y_{q,n}(w))_{n\ge 1}$, $q\in K$, we are going to use the argument developed to get Theorem 2 in \cite{Biggins}. 

For $z\in V_K$ and $w\in\mathscr{A}^*$ let 
$$
Y_{z,n}(w)=\sum_{v\in\mathscr{A}^n} \prod_{k=1}^n W_{z,v_{k}}(w\cdot v||w|+k-1)
$$
and denote $Y_{z,n}(\emptyset)$ by $Y_{z,n}$. Applying Proposition~\ref{VBE} to $\{V(w)=W_z(w)\}_{w\in\mathscr{A}^*}$ yields  for $n\ge 1$ 
$$
 \mathbb{E}(|Y_{z,n}-Y_{z,n-1}|^{p_K})\le C_{p_K}  M(z,p_K)^{n}\le C_{p_K} (\sup_{z\in V_K}M(z,p_K))^{n}, 
$$
where $Y_{z.0}=1$. Since, with probability 1, the functions $z\in V\mapsto Y_{z,n}$, $n\ge 0$, are analytic, if we fix a closed disc $D(z_0, 2\rho)$ included in $V$, the Cauchy formula yields 
$\sup_{z\in D(z_0,\rho)} |Y_{z,n}-Y_{z,n-1}|\le \rho^{-1}\int_{\partial D(z_0,2\rho)}|Y_{u,n}-Y_{u,n-1}|\,| \text{d}u |/2\pi$, so by using Jensen's inequality an then Fubini's Theorem we get
\begin{eqnarray*}
\mathbb{E}\big (\sup_{z\in D(z_0,\rho)} |Y_{z,n}-Y_{z,n-1}|^{p_K}\big )&\le&  2^{p_K}\int_{0}^{2\pi}\mathbb{E}(|Y_{z_0+2\rho e^{it},n}-Y_{z_0+2\rho e^{it},n-1}|^{p_K})\, \frac{\text{d}t}{2\pi} \\
&\le& 2 ^{p_K}C_{p_K} (\sup_{z\in V}\Phi(z,p_K))^{n}.
\end{eqnarray*}
This implies that, with probability 1,  $z\mapsto Y_{z,n}$ converges uniformly over the compact $D(z_0,\rho)$ to a limit $Y_z$. This also implies that $\|\sup_{z\in D(z_0,\rho)} Y_z\|_{p_K}<\infty$. Since $K$ can be covered by finitely many such discs, we get both the simultaneous convergence of $(Y_{q,n})_{n\ge 1}$ to $Y_q$ for all $q\in K$ and (2). Moreover, since $\Inte(J)$ can be covered by a countable increasing union of compact subintervals, we get the simultaneous convergence for all $q\in \Inte(J)$. The same holds simultaneously for all the functions $q\in \Inte(J)\mapsto Y_{q,n}(w)$, $w\in\mathscr{A}^*$, because $\mathscr{A}^*$ is countable. 

To finish the proof of (1) we need to establish that, with probability 1, $q\in K\mapsto Y_q$ does not vanish.  Up to an affine transform, we can suppose that $K=[0,1]$. If $I$ is a closed dyadic subinterval of $[0,1]$, we denote by $E_I$ the event $\{\exists \ q\in I:\ Y_q=0\}$, and by $I_0$ and $I_1$ its two sons. At first, we note that since for each fixed $q\in K$ a component of $W_q$ vanishes if and only the same component of $W$ vanishes too, each $E_I$ is a tail event. Consequently, if $I$ is a closed dyadic subinterval of $[0,1]$ and $\mathbb{P}(E_I)=1$, then $\mathbb{P}(E_{I_j})=1$ for some $j\in \{0,1\}$. Suppose that $\mathbb{P}(E_{[0,1]})=1$. The previous remark yields a decreasing sequence $(I(n))_{n\ge 0}$ of nested closed dyadic intervals such that $\mathbb{P}(E_{I(n)})=1$. Let $q_0$ be the unique element of $\bigcap I(n)$. Since $q\mapsto Y_q$ is continuous, we have $\mathbb{P}(Y_{q_0}=0)=1$. This contradicts the fact that the martingale $(Y_{q_0,n})_{n\ge 1}$ converges to $Y_{q_0}$ in $L^{p_K}$ norm.

\smallskip

\noindent
$\bullet$ {\it Proof of (3).} This is a simple consequence of the fact that by construction we have for all $n\ge 1$ and $w\in\mathscr{A}^*$ the branching property
$$
Y_{q,n+1}(w)=\sum_{k=0}^{b-1}W_{q,i}(w)\,Y_{q,n}(w\cdot i).
$$

\medskip

\noindent
{\it Proof of Proposition~\ref{keypro}.} (1) We simply denote $Q_W$ by $Q$ and we define
$$
\xi(q)=-\frac{\partial \Phi}{\partial q} (q,\tau(q))=-\mathbb{E}\Big(\sum_{i=0}^{b-1}\mathbf{1}_{\{W_i\neq 0\}}|W_i|^qL_i^{-\tau(q)}\log (|W_i|)\Big). 
$$
If $\varepsilon>0$, $n\geq 1$ and $\gamma\in \{-1,0,1\}$ we define
$$E_{q,n,\varepsilon}^1(\gamma)=\{t\in \mathscr{A}^{\mathbb{N}_+}: Q((t|_{n})^\gamma)\neq 0 \text{ and } e^{n(\xi(q)-\varepsilon)}|Q((t|_{n})^\gamma)|\geq 1\},$$
$$E_{q,n,\varepsilon}^{-1}(\gamma)=\{t\in \mathscr{A}^{\mathbb{N}_+}: Q((t|_{n})^\gamma)\neq 0 \text{ and } e^{n(\xi(q)+\varepsilon)}|Q((t|_{n})^\gamma)|\leq 1\}.$$

Our goal is to prove that for any compact subinterval $K$ of $\Inte(J)$ and $\varepsilon>0$, 
\begin{equation}\label{finite}
\mathbb{E}\Big (\sup_{q\in K}\sum_{n\geq 1}\mu_q(E_{q,n,\varepsilon}^\lambda(\gamma)\Big )<\infty 
\end{equation}
for all $\lambda\in \{-1,1\}$ and $\gamma\in\{-1,0,1\}$. Then, with probability 1, for all $q\in K$, $\lambda\in \{-1,1\}$ and $\gamma\in\{-1,0,1\}$, the series $\sum_{n\geq 1}\mu_q(E_{q,n,\varepsilon}^\lambda(\gamma))$ is finite. Since $\Inte(J)$ can be written as a countable union of compact subintervals, this holds in fact for all $q\in \Inte(J)$. Consequently, from the Borel-Cantelli lemma applied to $\mu_q/\|\mu_q\|$ we deduce that, with probability 1, for all $q\in \Inte(J)$, for $\mu_q$-almost every $t\in\mathscr{A}^{\mathbb{N}_+}$, there exists $N\ge 1$ such that for all $n\geq N$ and $\gamma\in \{-1,0,1\}$
$$
\text{ either }Q((t|_{n})^\gamma)=0\text{ or }|Q((t|_{n})^\gamma)|\in[e^{-n(\xi(q)+\varepsilon)},e^{-n(\xi(q)-\varepsilon)}].$$
Notice that when $t\in \supp(\mu_q)$, we have $Q((t|_{n})^0)=Q(t|_{n})\neq 0$. Consequently, with probability 1, for all $q\in K$, for $\mu_q$-almost every $t$,
\begin{eqnarray*}
&&\lim_{n\to\infty} \frac{\log |Q(t|_{n})|}{-n}\in [\xi(q)-\varepsilon,\xi(q)+\varepsilon],\\
&& \lim_{n\to\infty} \frac{\log |Q((t|_{n})^\gamma)|}{-n}\in\{+\infty\}\cup [ \xi(q)-\varepsilon,\xi(q)+\varepsilon]\text{ for } \gamma\in \{-1,1\}.
\end{eqnarray*}
Since this holds for a sequence of positive $\varepsilon$ tending to 0,  we have the desired result.

\smallskip

Now we prove ($\ref{finite}$).  Fix $K$, a non-trivial compact subinterval of $\Inte(J)$. For $\eta\geq 0$, $q\in K$ and $\gamma\in \{-1,0,1\}$,  by using a Markov inequality we get 
\begin{eqnarray*}\mu_q(E_{q,n,\varepsilon}^1(\gamma))&\leq &\sum_{w\in \mathscr{A}^n} \mu_q([w])  \textbf{1}_{\{Q( w^\gamma\neq 0\}} \Big(e^{n(\xi(q)-\varepsilon)}\cdot|Q(w^\gamma)|\Big)^\eta,\\
\mu_q(E_{q,n,\varepsilon}^{-1}(\gamma))&\leq& \sum_{w\in \mathscr{A}^n} \mu_q([w]) \textbf{1}_{\{Q(w^\gamma)\neq 0\}} \Big(e^{n(\xi(q)+\varepsilon)}\cdot|Q(w^\gamma)|\Big)^{-\eta}.
\end{eqnarray*}

Since $\mu_q([w])=\textbf{1}_{\{Q_(w)\neq 0\}} |Q_q(w)|Y_q(w)$, for $\lambda\in\{-1,1\}$ and $\gamma\in \{-1,0,1\}$ we get
$$\mu_q(E_{q,n,\varepsilon}^\lambda)\leq \sum_{w\in \mathscr{A}^n} e^{n(\lambda\eta\xi(q)-\varepsilon\eta)} \textbf{1}_{\{Q(w),Q(w^\gamma)\neq 0\}} Q_q(w)|Q(w^\gamma)|^{\lambda\eta} Y_q(w).$$

Now define
\begin{equation}\label{H}
H^{\eta,\lambda}_{n}(q,\gamma)=\sum_{w\in \mathscr{A}^n} e^{n(\lambda\eta\xi(q)-\varepsilon\eta)} \textbf{1}_{\{Q(w),Q(w^\gamma)\neq 0\}} Q_q(w)|Q(w^\gamma)|^{\lambda\eta}.
\end{equation}

Write $K=[q_0,q_1]$.  It follows from the independence between $\big \{H^{\eta,\lambda}_{n}(q,\gamma)\big\}_{q\in K}$ and $\{Y_q(w)\}_{w\in \mathscr{A}^n,q\in K}$ that for $n\geq 1$
\begin{eqnarray*}
\mathbb{E}(\sup_{q\in K}\mu_q(E_{q,n,\varepsilon}^\lambda(\gamma)) &\leq&\mathbb{E}(\sup_{q\in K} Y_q)\mathbb{E}\big (\sup_{q\in K}H^{\eta,\lambda}_{n}(q,\gamma)\big )\\
&\le & \mathbb{E}(\sup_{q\in K} Y_q)\Big( \mathbb{E}\big (H^{\eta,\lambda}_{n}(q_0,\gamma)\big )  + \int_{q_0}^{q_1} \mathbb{E}\big (\big|\frac{\text{d}}{\text{d}q}H^{\eta,\lambda}_{n}(q,\gamma)\big|\big ) \text{d} q\Big).
\end{eqnarray*}
\begin{lemma}\label{keylemma}
Let $\lambda\in\{-1,1\}$ and $\gamma\in \{-1,0,1\}$. There exist constants $C$, $\delta>0$ and $\eta_*>0$ such that for any $q\in K$, $\eta\in (0,\eta_*)$, $\lambda\in\{-1,1\}$ and $n\geq 1$,
$$\max\Big \{\mathbb{E}\big (H^{\eta,\lambda}_{n}(q,\gamma)\big ),\mathbb{E}\big (\big|\frac{\text{d}}{\text{d}q} H^{\eta,\lambda}_{n}(q,\gamma)\big|\big )\Big \}\leq C n e^{-n\delta}.$$
\end{lemma}

Then ($\ref{finite}$) comes from the fact that $ \mathbb{E}(\sup_{q\in K} Y_q)<\infty$ (see Proposition~\ref{muq}(2)). 

\medskip

\noindent
{\it{Proof of Lemma $\ref{keylemma}$}.} Recall that $\xi(q)=-\frac{\partial\Phi}{\partial q} (q,-\tau(q))$. Since $\Phi$ is twice continuously differentiable, we can chose $\eta_0>0$ such that for $\eta\in (0,\eta_0)$,
\begin{equation}\label{delta}
\delta_\eta=\inf_{q\in K}\varepsilon\eta-\lambda\eta \xi(q)-\log \big (\Phi(q+\lambda\eta,-\tau(q))\Big )>0.
\end{equation}

We now distinguish the cases $\gamma=0$ and $\gamma\in \{-1,1\}$. 

\smallskip

\noindent
$\bullet$ {\it The case $\gamma=0$.}  Straighforward computations using the definition of $H^{\eta,\lambda}_{n}(q,0)$ and taking into account the independence in the $b$-adic cascade construction yield a constant $C_K$ such that for all $q\in K$ and $n\ge 1$
\begin{eqnarray*}
\label{H0}\mathbb{E}\big (H^{\eta,\lambda}_{n}(q,0)\big )& =&\Phi(q+\lambda\eta,-\tau(q))^n e^{n(\lambda\eta \xi(q)-\varepsilon\eta)}\le e^{-n\delta_\eta}\\
\label{dH0}\mathbb{E}\big (\big|\frac{\text{d}}{\text{d}q} H^{\eta,\lambda}_{n}(q,0)\big|\big )&\le & C_K n \Phi(q+\lambda\eta,-\tau(q))^n e^{n(\lambda\eta \xi(q)-\varepsilon\eta )}\le C_Kne^{-n\delta_\eta}.
\end{eqnarray*}

\smallskip

\noindent
$\bullet$ {\it The case $\gamma=-1$.} For $n\ge 1$ we have 
\begin{equation}\label{voisin}
\bigcup_{w\in \mathscr{A}^n} (w^-,w) =\bigcup_{m=0}^{n-1} \bigcup_{u\in A^m}\bigcup_{i=0}^{b-2} (u\cdot i\cdot g_{n-1-m},u\cdot (i+1)\cdot d_{n-1-m})
\end{equation}
where $g_n$ (resp. $d_n$) is the word consisting of $n$ times the letter $b-1$ (resp. $0$). If $w=u \cdot (i+1) \cdot d_{n-1-m}$ and $w^-=u \cdot i\cdot g_{n-1-m}$ with $u\in \mathscr{A}^{\mathbb{N}_+}_m$ and $Q(w)Q(w^-)\neq 0$ then
\begin{equation*}\label{H-1}
\quad  Q_q(w)|Q(w^-)|^{\lambda\eta}=Q_q(u)Q(u)^{\lambda\eta}W_{q,i}(u)|W_{i+1}(u)|^{\lambda\eta}\prod_{k=m+1}^{n-1}W_{q,0}(w|_{k})|W_{b-1}(w^-|_{k})|^{\lambda\eta}.
 \end{equation*}
Again, simple computations yield $C_K>0$ such that for all $q\in K$, $n\ge 1$ and $\eta\in(0,\eta_0)$ we have  
$$
\max\big (\mathbb{E}(H^{\eta,\lambda}_{n}(q,\gamma), \mathbb{E}\big (\big|\frac{\text{d}}{\text{d}q} H^{\eta,\lambda}_{n}(q,\gamma)\big|\big )\big )\le C_K n\big (\Phi(q+\lambda\eta,\tau(q))e^{\lambda\eta \xi(q)-\varepsilon\eta}\big )^nS_n(q,\eta),
$$
where
$$
S_n(q,\eta)=\sum_{m=0}^{n-1}\Big [\frac{\mathbb{E}(W_{q,0})\mathbb{E}(|W_{b-1}|^{\lambda\eta})}{\Phi(q+\lambda\eta,\tau(q))}\Big ]^{m}.
$$

Due to (\ref{delta}), it is now enough to show that $S_n(q,\eta)$ is uniformly bounded with respect to $n$, $q\in K$ and $\eta$ if $\eta_0$ is small enough. This is due to the fact that the mapping $(q,r)\mapsto  \mathbb{E}(W_{q,0})\mathbb{E}(|W_{b-1}|^{r})/\Phi(q+r,\tau(q))$ is continuous in a neighborhood of $J\times \{0\}$ and by definition of $W_q$ and $\Phi$ it takes values less than 1 at points of the form $(q,0)$.  
 \smallskip
 
 \noindent
$\bullet$ {\it The case $\gamma=1$.} It uses the same ideas as the case $\gamma=-1$. 

\medskip

\noindent
(2) The proof is similar to the proof of (1). The only difference is that the components of $L$ are positive so the limit of $\frac{\log Q_L((t|_{n})^\gamma)}{-n}$ cannot be infinite. 
\medskip

\noindent
(3) We denote $Z_U^{(m)}$ by $Z$. Fix $K$ a non-trivial compact subinterval of $\Inte(J)$, $\lambda\in\{1,-1\}$ and  $\gamma\in\{-1,0,1\}$. For $a>1$ and $n\geq 0$ let
$$E_{n,a}^{\lambda}(\gamma)=\{t\in\mathscr{A}^{\mathbb{N}_+}: (Z(t|_{n})^\gamma))^{\lambda}>a^n\}.
$$

It is enough that we show that 
\begin{equation}\label{BC}
\mathbb{E}(\sup_{q\in K}\sum_{n\geq 0} \mu_q(E_{n,a}^{\lambda})(\gamma)) <\infty.
\end{equation}
Indeed, this implies that, with probability 1, for all $q\in K$, for $\mu_q$-almost every $t$, if $n$ is large enough then

$$-\log a \leq \liminf_{n\rightarrow\infty}\frac{\log Z(({t|_{n}})^\gamma)}{n}\leq \limsup_{n\rightarrow\infty}\frac{\log Z((t|_{n})^{\gamma})}{n}\leq \log a.$$
Since this holds for a sequence of numbers $a$ tending to 1, we have the conclusion. 

We have
\begin{eqnarray*}
\sup_{q\in K} \mu_q(E_{n,a}^{\lambda}(\gamma)) &=& \sup_{q\in K} \sum_{w\in \mathscr{A}^n} \textbf{1}_{\{Z(w^\gamma)^{\lambda}>a^n\}}\cdot \mu_q([w])\\
&=&\sup_{q\in K} \sum_{w\in \mathscr{A}^n} Q_{q,n}(w) \cdot   \textbf{1}_{\{Z(w^\gamma)^{\lambda}>a^n\}}\cdot Y_{q}(w).
\end{eqnarray*}
By using the independence between $\sigma(Q(w): w\in \mathscr{A}^n)$ and $\sigma(Z(w^\gamma),Y_q(w): w\in \mathscr{A}^n, q\in K)$, as well as the equidistribution of the random variables $ \textbf{1}_{\{Z(w^\gamma)^{\lambda}>a^n\}}\cdot Y_{q}(w)$, we get 

$$
\mathbb{E}\big (\sup_{q\in K} \mu_q(E_{n,a}^{\lambda}(\gamma))\big )\le \mathbb{E}\big (\textbf{1}_{\{Z(w_0^\gamma)^{\lambda}>a^n\}}\cdot \sup_{q\in K} Y_{q}(w_0)\big ) \mathbb{E}\big (\sup_{q\in K} H^{0,0}_n(q,\gamma)\big ),$$
where $H^{0,0}_n(q,\gamma)$ is defined as in (\ref{H}) and $w_0$ is any element of $\mathscr{A}^n$ such that $w_0^{\gamma}$ is defined. We learn from our computations in proving (1) that there exists a positive number $C_K$ such that $\mathbb{E}\big (\sup_{q\in K} H^{0,0}_n(q,\gamma)\big )\le C(1+|K|) n$. Moreover, the H\"older inequality yields 
\begin{eqnarray*}
\mathbb{E}\big (\textbf{1}_{\{Z(w_0^\gamma)^{\lambda}>a^n\}}\cdot \sup_{q\in K} Y_{q}(w_0)\big ) &\le& \| \sup_{q\in K} Y_{q}\|_{p_K}\mathbb{P}(Z^{\lambda\varepsilon}>a^{\varepsilon n})^{1-1/p_K}\\
&\le& \| \sup_{q\in K} Y_{q}\|_{p_K} [\mathbb{E}(Z^{\lambda\varepsilon })]^{1-1/p_K} a^{-n\varepsilon  (1-1/p_K)},
\end{eqnarray*}
where $\varepsilon>0$ is chosen such that  $\mathbb{E}(Z^{\lambda\varepsilon})<\infty$ (this is possible thanks to Proposition~\ref{xmoments}). Finally, $\mathbb{E}\big (\sup_{q\in K} \mu_q(E_{n,a}^{\lambda}(\gamma))\big )=O(n a^{-n\varepsilon (1-1/p_K)})$ (with $a>1$), hence (\ref{BC}) holds. 

\medskip

\noindent
(4) Fix $K$ a non-trivial compact subinterval of $\Inte(J)$. For $a>1$, $n\geq 0$ and $q\in K$ let
$$
F_{n,a}^{\lambda}(q)=\{t\in\mathscr{A}^{\mathbb{N}_+}: Y_q(t|_{n})>a^n\}.
$$
For $\eta>0$, we have 
\begin{eqnarray*}
\sup_{q\in K} \mu_q(F_{n,a}(q)) &=& \sup_{q\in K} \sum_{w\in \mathscr{A}^n} \textbf{1}_{\{Y_q(w)>a^n\}}\cdot \mu_q([w])\\
&\le &\sup_{q\in K} Q_{q,n}(w)|^q \cdot a^{-n\eta} Y_q(w)^{1+\eta}. 
\end{eqnarray*}
Consequently, taking $\eta=p_K-1$ and using the same kind of estimations as in the proof of (3) we obtain

$$
\mathbb{E}\big (\sup_{q\in K} \mu_q(F_{n,a}(q))\big )\le  a^{-n(p_K-1)}\mathbb{E}\big (\sup_{q\in K} H^{0,0}_n(q,\gamma)\big ) \mathbb{E}(\sup_{q\in K}Y_q^{p_K})=O(n a^{-n(p_K-1)}),
$$
hence the result.

\subsection{Proofs of results of Section \ref{lowerbound2}}
We only deal with the case $h=\underline h$, the case  $h=\bar h$ being similar. Then $q=\overline q>0$. 

\subsubsection{The case $q<\infty$}$\ $\\

\noindent
{\it Proof of Proposition~\ref{endpoints}.} At first, we specify a subset $E_{h}$ of $E_F(h)$ of positive $\mu_{q}$-measure. For $N\ge 1$, let $E_h(N)=\{t\in \mathscr{A}^{\mathbb{N}_+}: \forall n\geq N, \mu_{q}([t|_{n}])\leq 1\}$. With probability 1, there exists $N\ge 1$ such that $\mu_q(E_h(N))>0$, otherwise, $\mu_q$ is concentrated on a finite number of singletons with a positive probability, which is impossible since by construction $\supp(\mu_q)=\supp(F')$ and $\dim\, \supp(F')=-\varphi(0)>0$ almost surely. Thus, on a measurable set of probability 1, we can define the measurable function $N(\omega)=\inf\{N: \mu_q(E_h(N))>0\}$. Then we set $E_h(\omega)=E_h(N(\omega))$. 

\smallskip

\noindent
(1)  We denote $Q_W$ by $Q$ and $-\partial \Phi/\partial q(q,\tau(q))$ by $\xi(q)$. For $\varepsilon>0$, $\gamma\in\{-1,0,1\}$ and $n\geq 1$ we define
\begin{eqnarray*}
E_{n,\varepsilon}^1(\gamma)&=&\{t\in E_{ h}: Q(t|_{n})\neq 0 \text{ and } e^{n(\xi(q)-\varepsilon)}|Q((t|_{n})^\gamma)|\geq 1\},\\
E_{n,\varepsilon}^{-1}(\gamma)&=&\{t\in E_{ h}: Q(t|_{n})\neq 0 \text{ and } e^{n(\xi(q)+\varepsilon)}|Q((t|_{n})^\gamma)|)\leq 1\}\}.
\end{eqnarray*}

The result will follow if we show that for any $\varepsilon>0$ and $\lambda\in \{-1,0,1\}$, with probability 1,
\begin{eqnarray}
\label{S1}\sum_{n\ge 1}\mu_q(E_{n,\varepsilon}^\lambda(\gamma))<\infty.
\end{eqnarray}

We deal with the case $\gamma=0$.  Let $ \theta$ and $ \eta$ be two numbers in $(0,1]$, that will be specified later.  
By using a Markov inequality and the definition of $\mu_h$ we can get
\begin{eqnarray*}
\mu_q(E_{n,\varepsilon}^\lambda(0))&\le & \sum_{\substack{w\in\mathscr{A}^n\\ \mu_q([w])\le 1}} \mu_{q}([w])\textbf{1}_{\{Q(w)\neq 0\}} \big(e^{n(\xi(q)-\lambda\varepsilon)}\cdot Q(w)\big )^{\lambda\eta}\\
&\le &  \sum_{\substack{w\in\mathscr{A}^n\\ \mu_q([w])\le 1}} \mu_{q}([w])^{\theta}\textbf{1}_{\{Q(w)\neq 0\}} \big(e^{n(\xi(q)-\lambda\varepsilon)}\cdot Q(w)\big )^{\lambda\eta}= S_{n,\varepsilon}^{\lambda}(\theta, \eta),
\end{eqnarray*}
where
$$
S_{n,\varepsilon}^{\lambda}( \theta, \eta)=\sum_{w\in \mathscr{A}^n}e^{n(\lambda\eta\xi(q)-\varepsilon\eta)} \textbf{1}_{\{Q(w)\neq 0\}}Q_q(w)^{\theta}Q(w)^{\lambda\eta}Y_q^{\theta}.
$$
Consequently, (\ref{S1}) will follow if we show that 
\begin{equation}\label{S2}
\sum_{n\ge 1}\mathbb{E}(S_{n,\varepsilon}^{\lambda}(\theta,\eta))<\infty.
\end{equation} 
We have 
$$
\mathbb{E}(S_{n,\varepsilon}^{\lambda}( \theta, \eta))=\mathbb{E}(Y_q^{\theta})e^{n(\lambda\eta\xi(q)-\varepsilon\eta)} \Phi(\theta q+\lambda\eta,\theta \tau(q))^n.
$$
Let $\widetilde \xi(q)=\frac{\partial \Phi}{\partial t}(q,\tau(q))$. By definition of $\xi(q)$, $\widetilde \xi(q)$ and $-\tau(q)$ we have 
\begin{eqnarray*}
&&\lambda\eta\xi(q)+\log \Phi(\theta q+\lambda\eta,\theta\tau(q))\\
&=&-\xi(q)q(\theta-1)+ \widetilde \xi(q) (\theta-1) \tau(q)+O( [q(\theta-1)+\lambda\eta]^2)
\end{eqnarray*}
as $(\theta,\eta)\to (1^-,0)$. Moreover, we have
$$-\xi(q)q+ \widetilde \xi(q) \tau(q)= -\widetilde \xi(q)(\tau'(q)q-\tau(q)) = 0.$$
It follows that if we fix $\eta$ small enough and $\theta$ close enough to $1^-$ we have
$$e^{n(\lambda\eta\xi(q)-\varepsilon\eta)} \Phi(\theta q+\lambda\eta,\theta\tau(q))^n\le e^{-n\varepsilon\eta/2}.$$
Since $\mathbb{E}(Y_q^{\theta})<\infty$, we get \eqref{S2}.

In the case $\gamma=-1$, we leave the reader check that like in the proof of Proposition~\ref{keypro}(1)  we can find a constant $C>0$ such that for $\theta,\eta\in (0,1]$ we have  
\begin{eqnarray*}
&&\mu_q(E_{n,\varepsilon}^\lambda(-1)) \\
&\le& C\cdot \mathbb{E}(Y_q^{\theta})e^{n(\lambda\eta\xi(q)-\varepsilon\eta)} \Phi(\theta q+\lambda\eta,\theta\tau(q))^n \sum_{m=0}^{n-1}\Big [\frac{\mathbb{E}(W_{q,0}^\theta)\mathbb{E}(|W_{b-1}|^{\lambda\eta})}{\Phi(\theta q+\lambda\eta,\theta\tau(q))}\Big ]^{m}.
\end{eqnarray*}

\smallskip

\noindent
(2) The proof is similar to that of (1).

\smallskip

\noindent
(3) We denote $Z^{(m)}_U$ by $Z$. Let $\theta,\eta\in (0,1]$. For $n\ge 1$, $\gamma\in\{-1,0,1\}$, $\lambda\in\{-1,1\}$ and $\varepsilon>0$ let $E^\lambda_{n,\varepsilon}(\gamma)= \{t\in E_h: Z((t|_{n})^\gamma)^\lambda> e^{n\varepsilon}\}.$
We have 
\begin{eqnarray*}
\mu_q\big (E^\lambda_{n,\varepsilon}(\gamma))\}\big )
\le \sum_{\substack{w\in\mathscr{A}^n\\ \mu_q([w])\le 1}} \mu_q([w])^\theta e^{-n\varepsilon\eta} Z((t|_{n})^\gamma)^{\lambda\eta}.
\end{eqnarray*}
Thus, 
\begin{eqnarray*}
\mathbb{E}\Big (\mu_q\big (E^\lambda_{n,\varepsilon}(\gamma)\big )\Big )
&\le & e^{-n\varepsilon\eta} \Phi(\theta q,\theta\tau(q))^n \mathbb{E}( Y_q(w)^\theta Z(w^\gamma)^{\lambda\eta}).
\end{eqnarray*}
If we choose $\theta =1-\eta^{2/3}$, then by using H\"older's inequality we get,  
$$
\mathbb{E}(Y_q^{\theta}Z(w)^{\lambda\eta})\le \|Y_q^{\theta}\|_{ 1+\eta^{3/4}}\big (\mathbb{E}(Z^{\lambda (1+\eta^{3/4})\eta^{1/4})}\big )^{\eta^{3/4}/(1+\eta^{3/4})}<\infty$$
 since $\mathbb{E}(Y_q^\theta)<\infty$ for $\theta \in (0,1)$ and $\mathbb{E}(Z^{\lambda \beta})<\infty$ if $|\beta|$ is small enough (see Proposition~\ref{xmoments}(2)). Moreover, by definition of $\tau(q)$, since the $L_i$ are smaller than 1, we have $\Phi(\theta q,\theta\tau(q))
<1$. Consequently, $\sum_{n\ge 1} \mathbb{E}\big (\mu (E^\lambda_{n,\varepsilon}(\gamma) )\big ) <\infty$. We conclude as in the proof of Proposition~\ref{keypro}(3). 
  
\medskip

\noindent
\subsubsection{The case $q=\infty$}$\ $\\

\noindent
{\it Proof of Proposition~\ref{qk}.} First we have the following lemma. 
\begin{lemma}\label{nondegeneration}
If $\displaystyle \sum_{n\geq 1} \prod_{k=0}^{n-1}\Phi\big (2 q_k,2\tau(q_k) \big )^{1/2}<\infty$ then  for every $w\in\mathscr{A}^*$, $Y_{q,n}(w)$ converges to $Y_q(w)$ in $L^1$ norm as $n\to \infty$; in particular $\mathbb{E}(Y_q(w))=1$. 
\end{lemma}

\begin{proof} An application of Proposition~\ref{VBE} to $Y_{q,n}-Y_{q,n-1}$ and $p=2$ yields 
\begin{equation}\label{L1bound}
\sum_{n\geq 1} \|Y_{q,n}-Y_{q,n-1}\|_{2}\leq C \sum_{n\geq 1}\prod_{k=0}^{n-1}\Phi\big (2 q_k,2\tau(q_k) \big )^{1/2}.
\end{equation}
where we set $Y_{q,0}=1$ and $C$ is the supremum of the constants $C_p$ invoked in Proposition~\ref{VBE}. 
Then, since $\|Y_{q.n}-Y_{q,n-1}\|_1\le \|Y_{q,n}-Y_{q,n-1}\|_{2}$ we have the conclusion for $w=\emptyset$. Now, if $m\ge 1$ we have 
\begin{equation}\label{identn}
Y_q=\sum_{w\in\mathscr{A}^{\mathbb{N}_+}_m} Q_q(w)\,Y_q(w),
\end{equation}
where the random variables $Y_q(w)$, $w\in\mathscr{A}^m$, are identically distributed, as well as the discrete processes $(Y_{q,n}(w))_{n\ge 1}$ converging to them. Consequently, if $Y_q(w)$ is not the limit of $(Y_{q,n}(w))_{n\ge 1}$ in $L^1$ for some $w\in \mathscr{A}^m$, then $\mathbb{E}(Y_q(w))<1$ and the same holds for all $w\in\mathscr{A}^m$. In particular, (\ref{identn}) yields $\mathbb{E}(Y_q)<1$, which is in contradiction with the convergence in $L^1$ norm of $(Y_{q,n})_{n\ge 1}$. 

\end{proof}

Now we specify the sequence $(q_k)_{k\ge 0}$.  We discard the obvious case where $\tau$ is affine and assume that $\tau$ is strictly concave.
\medskip

The graph of the function $\tau$ has the asymptote line $l(q)=h q- \tau^*(h)$ with $\tau^*(h)\in [0,1)$. For $\delta \geq 0$, let $l_\delta(q)= l(q) -\delta$. We deduce from the strict concavity of $\tau$  that for any $\delta\in (0,-\tau(0)-\tau^*(h)]$ there is a unique $q(\delta)>0$ such that $ l_\delta(q(\delta))=\tau(q(\delta))$. Moreover, $\delta\mapsto q(\delta)$ is continuous and strictly decreasing, and  $q(\delta)\to \infty$ as $\delta\to 0$. Fix $k_0$ such that $\frac{1}{\log k_0}\in (0,-\tau(0)-\tau^*(h))$, and for $k\ge 0$ let  $\delta_k=1/\log (k_0+k)$. Then choose $q_k=q(\delta_k)$ for $k\ge 0$. By using the definition of $l_\delta$ and $\delta(\cdot)$ as well as the concavity of $\tau$, we obtain for all $k\ge 0$ 
\begin{equation}\label{AAA}
\varepsilon_k=\tau(2q_{k})-2\tau(q_{k}) =\tau(2q_k)-2l_{\delta_k}(q_k)\ge l_{\delta_k} (2q_k)- 2l_{\delta_k}(q_k)=\tau^*(h)+\delta_k.
\end{equation}
We also have for any conjugate pair $(\alpha,\alpha')$ such that $1/\alpha+1/\alpha'=1$
\begin{eqnarray*}
\Phi\big (2 q_k, 2\tau(q_k) \big )&=&\mathbb{E}\Big (\sum_{i=0}^{b-1}W_{2q_k,i}L_i^{\varepsilon_k}\Big )\\
&\le &\Big (\mathbb{E}\Big (\sum_{i=0}^{b-1}W_{2q_k,i}^{\alpha'}\Big )\Big )^{1/\alpha'}\Big (\mathbb{E}\Big (\sum_{i=0}^{b-1}L_i^{\varepsilon_k\alpha}\Big)\Big )^{1/\alpha}.
\end{eqnarray*}
Our assumption that $\tau$ as an asymptote at $\infty$ implies that $\tau(q)/q$ is increasing, so $\mathbb{E}\Big (\sum_{i=0}^{b-1}W_{q',i}^{\alpha'}\Big )=\Phi(\alpha' q',\alpha'\tau(q'))\le \Phi(\alpha' q',\tau(\alpha' q'))=1$ for all $q'>0$ and $\alpha'>1$. Also, the fact that the $L_i$ belong to $(0,1)$ implies that $\varphi_L(q')\sim \bar{a}q'+c$ at $\infty$ with $\bar{a}>0$, so by choosing $\alpha$ large enough we ensure  $ \mathbb{E}\Big (\sum_{i=0}^{b-1}L_i^{\varepsilon_k\alpha}\Big)=b^{-\varphi_L(\varepsilon_k\alpha)}\le b^{-\bar{a} \varepsilon_k\alpha/2}$. Thus $\Phi\big (2 q_k,2\tau(q_k) \big )\le b^{-\bar{a} \varepsilon_k/2}$. Consequently, for all $n\ge 1$, we have 
$$
\prod_{k=0}^{n-1}\Phi\big (2 q_k,2\tau(q_k) \big )^{1/2}\le b^{-\bar{a}S_n/4},
$$
where $S_n=\displaystyle \sum_{k=0}^{n-1}\varepsilon_k$. We have either $\tau^*(h)>0$ and $S_n\ge n\tau^*(h)/2$ or $\tau^*(h)=0$ and there exists $n_0\ge 1$ such that $S_n \ge n/(4\log (n))$ for $n\ge n_0$. In both cases the conclusion of Lemma~\ref{nondegeneration} holds.

The same arguments as those used in the proof of Lemma~\ref{nondegeneration} show that for $w\in \mathscr{A}^*$ and $n\ge 1$ we have (with $Y_0(w)=1$)
\begin{equation}\label{IY}\|Y_{q,n}(w)-Y_{q,n-1}(w)\|_{2} \leq C\,  b^{-n \bar{a} (S_{|w|+n}-S_{|w|})/4}.
\end{equation}

If $\tau^*(h)>0$, by using (\ref{AAA}) we get $\|Y_{q,n}(w)-Y_{q,n-1}(w)\|_{2} \leq C b^{-n \bar{a}\tau^*(h)/4}$, an upper bound which does not depend on $w$, and finally $\sup_{w\in\mathscr{A}^*}\|Y_q(w)\|_2<\infty$.

\smallskip 

If $\tau^*(h)=0$, we have $\|Y_{q,n}(w)-Y_{q,n-1}(w)\|_{2}\le C \,  b^{\bar{a}S_{|w|}/4}\cdot b^{-\bar{a}S_{|w|+n}/4}$. Thus, 
\begin{eqnarray*}
\|Y_q(w)\|_2 &\le&  1+\sum_{n\ge 1}\|Y_{q,n}(w)-Y_{q,n-1}(w)\|_{2} \\
&\le& 1+C \cdot b^{\bar{a}S_{|w|}/4} \sum_{n\ge 1} b^{-\bar{a}S_{|w|+n}/4}\le 1+ C \cdot  b^{\bar{a}S_{|w|}/4} L,
\end{eqnarray*}
where $L=\sum_{n\ge 1} b^{-\bar{a}S_{n}/4}$ is finite because $S_n\ge n/4 \log (n)$ for $n$ large enough. 

Now, we notice that we also have by concavity of $\tau$ 
\begin{equation}\label{eq}
\tau(2q_{k})-2\tau(q_{k})=\tau (2q_k)-2l_{\delta_k}(q_k) \le  l(2q_k)-2l_{\delta_k}(q_k)=2\delta_k.
\end{equation}
Due to our choice for $\delta_k$, this implies that for $|w|$ large enough we have $S_{|w|}\le \frac{2|w|}{\log |w|}$. Finally, there exists $C'>0$ such that for all $w\in \mathscr{A}^*$ we have $\|Y_q(w)\|_2\le C' b^{\frac{\bar{a}|w|}{2\log |w|}}$.

\medskip

\noindent
{\it Proof of Proposition~\ref{endpoints2}.} We first prove the following proposition

\begin{proposition}\label{endpoints3}
Let $h\in\{\bar h, \underline h\}$ and $q\in\{-\infty,\infty\}$ such that $h=\lim_{J\ni q'\to q}\tau'(q)$. \begin{enumerate}
\item Let $m\ge 1$, $\gamma_1,\gamma_2\in\{-1,0,1\}$ and $\delta_1,\delta_2\in\{0,1\}$. With probability 1, for $\mu_{q}$-almost every $t\in \supp(\mu_q)$, $\displaystyle
\lim_{n\to\infty} \frac{\log |Q_W(t|_{n})| \cdot Z^{(m)}_W(t|_{n}^{\gamma_1})^{\delta_1}}{\log Q_L(t|_{n})\cdot Z_L^{(1)}(t|_{n}^{\gamma_2})^{\delta_2}}=h.$

\smallskip

\item
For $t\in\mathscr{A}^{\mathbb{N}_+}$, $i\in\{0,b-1\}$ and $n\ge 1$ we define $N_n(t)=\max\{0\le j\le n: \ \forall \ 0\le k\le j,\ t_{n-k}=i\}$. With probabilility 1, for $\mu_q$-almost every $t$ we have $N_n(t)=o(n)$. 

\smallskip

\item Let  $\varepsilon\in (0,1)$, $i\in\{0,b-1\}$ and $r\in\{0,...,b-1\}$. There exists $\alpha(\varepsilon)\in (0,\varepsilon)$ such that, with probability 1, for $\mu_q$-almost every $t$, for $n$ large enough and $n(1-\alpha(\varepsilon))\le p\le n-1$ such that $W_{n,p}(t|_{n}):=W_{r}(t|_{p-1})\prod_{k=p+1}^{n}W_i(t|_{k-1})\neq 0$, we have 
$$\log |W_{n,p}(t|_{n})|/\log O_L(t|_{n}) \ge -\varepsilon.$$
\end{enumerate}
\end{proposition}
{\it Proof.} (1) For $\gamma_1,\gamma_2\in\{-1,0,1\}$, $\delta_1,\delta_2\in\{0,1\}$, $m\ge 1$ and $w\in\mathscr{A}^*$ we simply denote
\[
\begin{cases}
O_W(w)=Q_W(w)\cdot Z_W^{(m)}(w^{\gamma_1})^{\delta_1}, \ O_L(w)=Q_L(w)\cdot Z_L^{(1)}(w^{\gamma_2})^{\delta_2},\\
Z(w)=Z_L^{(1)}(w^{\gamma_2})^{\delta_2}/Z_W^{(m)}(w^{\gamma_1})^{\delta_1}.
\end{cases}
\]
For $\varepsilon>0$, $n\geq 1$, and $\lambda\in \{-1,1\}$ we define
$$E^\lambda_{n,\varepsilon}=\{t\in \mathscr{A}^{\mathbb{N}_+}: O_W(t|_{n})\neq 0 \text{ and } O_W(t|_{n})^\lambda <O_L(t|_{n})^{\lambda h+ \varepsilon}\}.
$$
For any $\eta_n>0$ we have 
\begin{eqnarray*}
\mu_q(E^\lambda_{n,\varepsilon})&\le &\sum_{w\in\mathscr{A}^n}\mu_q(w) |Q_W(w)|^{-\lambda\eta_n}Q_L(w)^{\lambda \eta_n h+ \varepsilon\eta_n } Z(w)^{\lambda\eta_n}\\
&=&\sum_{w\in\mathscr{A}^n} Y_q(w)Z(w)^{\lambda\eta_n} \prod_{k=1}^{n} W_{w_{k}}(w|_{k-1})^{q_{k-1}-\lambda\eta_n}L_{w_k}(w|_{k-1})^{\lambda \eta_n h+ \varepsilon\eta_n-\tau(q_{k-1}) }.
\end{eqnarray*}
This  yields
\begin{equation}\label{inftyupp}
\mathbb{E}(\mu_q(E^\lambda_{n,\varepsilon})) \le \mathbb{E}(Y_q\cdot Z^{\lambda\eta_n})\prod_{k=1}^{n}\Phi(q_{k-1}-\lambda\eta_n,\tau(q_{k-1})-\lambda h\eta_n-\varepsilon\eta_n).
\end{equation}

Let us make the following observation. For any $q_k,\eta_n>0$ we can write
\begin{eqnarray}
&&\log\Phi(q_k-\lambda\eta_n,\tau(q_k)-\lambda h\eta_n-\varepsilon\eta_n)\nonumber\\
&=&\lambda\eta_n\xi(q_k) -(\lambda h\eta_n+\varepsilon\eta_n)\widetilde\xi(q_k) +\frac{1}{2}\Delta(\zeta_k) \eta_n^2 \nonumber \\
&=&-\varepsilon\eta_n \widetilde{\xi}(q_k) +\lambda\eta_n \widetilde{\xi}(q_k)(\tau'(q_k)-h)+\frac{1}{2}\Delta(\zeta_k)\eta_n^2,\label{asymphi}
\end{eqnarray}
where $\zeta_k=(\zeta_1,\zeta_2)=s(q_k,\tau(q_k))+(1-s)(q_k-\lambda\eta_n,\tau(q_k)-\lambda h\eta_n-\varepsilon\eta_n)$ for some $s\in[0,1]$, and
\begin{eqnarray*}
\Delta(\zeta_k)&=&\lambda^2\frac{\partial^2}{\partial q^2}\Phi(\zeta_k) +(\lambda h+\varepsilon)^2\frac{\partial^2}{\partial t^2}\Phi(\zeta_k)+2\lambda(\lambda h+\varepsilon)\frac{\partial^2}{\partial q \partial t}\Phi(\zeta_k).
\end{eqnarray*}
Also, we have 
\begin{eqnarray*}
&&\log\Phi(q_k-\lambda\eta_n,\tau(q_k)-\lambda h\eta_n-\varepsilon\eta_n)\\
&=&\log\Phi(q_k-\lambda\eta_n,\tau(q_k-\lambda\eta_n)+\lambda(\tau'(q_k)-h)\eta_n + \tau''(\widetilde q_k) \eta_n^2/2-\varepsilon\eta_n),
\end{eqnarray*}
where $\widetilde q_k\in [q_k-\lambda\eta_n,q_k]$. Since $\tau$ is concave and $\tau'(q)$ tends to $h$ at $\infty$, if $\eta_n$ is small enough, then for $k$ large enough we have
$$\lambda(\tau'(q_k)-h)\eta_n + \tau''(\widetilde q_k) \eta_n^2/2-\varepsilon\eta_n\le -\varepsilon\eta_n/2<0,$$
hence $\log\Phi(q_k-\lambda\eta_n,\tau(q_k)-\lambda h\eta_n-\varepsilon\eta_n)<\log\Phi(q_k-\lambda\eta_n,\tau(q_k-\lambda \eta_n))=0$. Hence, due to \eqref{asymphi}
\begin{equation}\label{asymphi2}
0\ge  -\varepsilon\eta_n \widetilde{\xi}(q_k) +\lambda\eta_n \widetilde{\xi}(q_k)(\tau'(q_k)-h)+\frac{1}{2}\Delta(\zeta_k)\cdot \eta_n^2.
\end{equation}
Moreover, under our assumptions, the multifractal analysis of the Mandelbrot  measure $\mu_L=F_L'$ achieved in~\cite{B2}) implies that for any random probability vector $\widetilde{W}$ with $\mathbb{E}(\sum_{i=0}^{b-1}\widetilde{W}_i)=1$ and $\mathbb{E}(\sum_{i=0}^{b-1}\widetilde{W}_i\log\widetilde{W}_i)<0$, $-\mathbb{E}(\sum_{i=0}^{b-1}\widetilde{W}_i\log_b L_i)$ is a H\"older exponent for $\mu_L$, so it must belong to $[\underline{a},\overline{a}]$, where 
$$0<\underline{a}=\lim_{q\to+\infty} \tau_{F_L}(q)/q \le \lim_{q\to-\infty} \tau_{F_L}(q)/q=\overline{a}<\infty.$$
Applying this with $\widetilde W=W_{q_k}$ yields $\widetilde{\xi}(q_k)/\log(b)\in[\underline{a},\overline{a}]$ for all $k\ge 0$. Also, since $q_k\nearrow \bar{q}=+\infty$ we have $\tau'(q_k)=\xi(q_k)/\widetilde{\xi}(q_k)\searrow h<\infty$, so $\sup_{k\ge0}\xi(q_k)<\infty$. These properties together with \eqref{asymphi2} yield $c=\sup_{k} \Delta(\zeta_k) <\infty$. 

For simplicity we define $\underline{a}:=\log (b)\,  \underline{a}$ and $\overline{a}:= \log (b)\, \overline{a}$.

By using again the fact that $\tau'(q_k)-h\ge 0$ as well as~\eqref{inftyupp} and~\eqref{asymphi} we get
\begin{equation}
\mathbb{E}(\mu_q(E^\lambda_{n,\varepsilon})) \le e^{-\varepsilon\underline{a}n\eta_n+\overline{a}(\sum_{k=1}^{n}\tau'(q_k)-h)\eta_n+c\eta_n^2/2} \cdot \mathbb{E}(Y_q(w)\cdot Z^{\lambda\eta_n}(w)).
\end{equation}

We take $\eta_n=1/\sqrt{ \log (n_m+n)}$ for all $n\ge 1$, where $n_m$ is an integer large enough so that for any $\lambda\in\{-1,1\}$ we have $\mathbb{E}((Z_W^{(m)})^{\frac{4\lambda}{\sqrt{\log (n_m+1)}}})<\infty$, as well as $\mathbb{E}((Z_L^{(1)})^{\frac{4\lambda}{\sqrt{\log (n_m+1)}}})<\infty$ (the existence of $n_m$ comes from Propositioin~\ref{xmoments}). Then due to Proposition~\ref{qk}, for $n\ge 1$ and $w\in\mathscr{A}^n$ we have
\begin{equation}\label{ineq1}
\mathbb{E}(Y_q(w)\cdot Z^{\lambda \eta_n})\le \|Y_q(w)\|_2 \cdot \|(Z_L^{(1)})^{\lambda \delta_2\eta_n}\|_4\cdot \|(Z_W^{(m)})^{-\lambda\delta_1 \eta_n}\|_4 =O( b^{2n/\log (n)}).
\end{equation}

Notice that  $\sum_{k=1}^{n}\tau'(q_k)-h=o(n)$, since $\tau'(q_k)- h\searrow 0$ when $k\to\infty$. Thus, due to our choice for $\eta_n$, for $n$ large enough we have  
\begin{equation}\label{ineq2}
-\varepsilon\underline{a}n\eta_n+\overline{a}(\sum_{k=1}^{n}\tau'(q_k)-h)\eta_n+c\eta_n^2=-\varepsilon\underline{a}n\eta_n+o(n\eta_n).
\end{equation}
Then (\ref{ineq1}) and  (\ref{ineq2}) together yield $\mathbb{E}(\mu_q(E^\lambda_{n,\varepsilon}))=O(b^{-\varepsilon\underline{a}n/2\sqrt{\log (n)}})$. 
\medskip

\noindent
(2) Let us recall that $\mathbb{E}(\|\mu_q\|)=\mathbb{E}(Y_q)=1$ and introduce on $\Omega\times \mathscr{A}^{\mathbb{N}_+}$ the ``Peyri\`ere" probability measure $\mathcal{Q}_q$ defined by 
$$
\mathcal{Q}_q(A)=\mathbb{E}\Big (\int \mathbf{1}_A(\omega,t)\mu_q ^\omega(\text{d}t)\, \mathbb{P}(\text{d}(\omega))\Big ), \ A\in \mathcal{B}\otimes \mathcal{S}.
$$
Notice that ``$\mathcal{Q}_q$-almost surely" means ``with probability 1, $\mu_q$-almost everywhere". 
\smallskip

\noindent
 Without loss of generality, we can assume that for $i\in\{0,b-1\}$ the sequence
 $$\big (\mathbb{E}(\mathbf{1}_{\{W_i\neq 0\}}|W_i|^{q_k} L_i^{-\tau(q_k)})\big )_{k\ge 0}$$
 has a limit $f_i$ as $k\to\infty$, since this sequence takes values in the bounded interval $[0,1]$. 

Now for $n\ge 1$, $i\in \{0,b-1\}$ and $(\omega,t)\in \Omega\times \mathscr{A}^{\mathbb{N}_+}$ set $f_{i,n}(\omega,t)=\mathbf{1}_{\{i\}}(t_n)$. It is not difficult to show that the random variables $f_{i,n}$, $n\ge 1$ are $\mathcal{Q}_q$ independent. Moreover, we have
$$
\widetilde f_{i,n}:=\mathbb{E}_{\mathcal{Q}_q}(f_{i,n})=\mathbb{E}_{\mathcal{Q}_q}(f_{i,n}^2)=\mathbb{E}(W_{q_{n-1},i})=\mathbb{E}(|W_i|^{q_{n-1}}L_i^{-\tau(q_{n-1})}).
$$
Indeed, 
\begin{eqnarray*}
&&\mathbb{E}_{\mathcal{Q}_q}(f_{i,n})\\
&=&\sum_{w\in \mathscr{A}^n,\ w_n=i} \mathbb{E}(\mu_q([w]) )=
\sum_{w\in \mathscr{A}^n,\ w_n=i}\mathbb{E}(Q_q(w|_{n-1}))\mathbb{E}(W_{q_{n-1},i}(w|_{n-1})) \mathbb{E}(Y_q(w)).
\end{eqnarray*}

Consequently, $\widetilde f_{i,n}$ converges to $f_i$ as $n\to\infty$, and on $(\Omega\times \mathscr{A}^{\mathbb{N}_+}, \mathcal{B}\otimes \mathcal{S},\mathcal{Q}_q)$, the martingale $\sum_{k=1}^n (f_{i,k}-\mathbb{E}_{\mathcal{Q}_q}(f_{i,k}))/k$ is bounded in $L^2$ norm by $\sum_{k\ge 1} \widetilde f_{i,k}(1-f_{i,k})/k^2$. It follows that the series $\sum_{k\ge 1} (f_{i,k}-\mathbb{E}_{\mathcal{Q}_h}(f_{i,k}))/k$ converges $\mathcal{Q}_q$-almost surely, and the Kronecker lemma implies that $\sum_{k=1}^nf_{i,n}/n$ converges to $f_i$ $\mathcal{Q}_q$-almost surely.  This implies (1).

\medskip

\noindent
(3) Let $\alpha\in (0,\varepsilon)$. Fix $\eta\in (0,1)$ and for $n$ large enough let $p\in [(1-\alpha)n,n-1]$ be an integer. For $(\omega,t)\in \Omega\times \mathscr{A}^{\mathbb{N}_+}$, let $X_{n,p}(\omega,t)=\log |W_{n,p}(t|_{n})|/\log O_L(t|_{n})$. We have 
\begin{eqnarray*}
\mathcal{Q}_q(X_{n,p}< -\varepsilon)&=&\sum_{w\in \mathscr{A}^n}\mathbb{E}(\mathbf{1}_{\{\log |W_{n,p}(t|_{n})|/\log O_L(t|_{n}) < -\varepsilon\}}\mu_q([w]))\\
&\le& \sum_{w\in \mathscr{A}^n} \mathbb{E}(Y_q(w)Z_L(w)^{\eta\varepsilon}) \mathbb{E}(Q_q(w|_{p-1})Q_L(w|_{p-1})^{\eta\varepsilon})\\
&&\qquad\cdot\ \mathbb{E}(W_{q_{p-1},w_p}(w|_{p-1}) L_{w_p}(w|_{p-1})^{\eta\varepsilon}|W_r(w|_{p-1})|^\eta)\\
&&\qquad\cdot \prod_{k=p+1}^{n} \mathbb{E}(W_{q_{k-1},w_k}(w|_{k-1}) L_{w_k}(w|_{k-1})^{\eta\varepsilon}|W_i(w|_{k-1})|^\eta)
\end{eqnarray*}
for any $\eta>0$. Applying the Cauchy-Schwarz inequality in the right hand side of the above inequality  yields
\begin{eqnarray*}
\mathcal{Q}_q(X_{n,p}< -\varepsilon) 
&\le& \prod_{k=1}^{p-1}\Phi(q_{k-1},\tau(q_{k-1})-\eta\varepsilon) \prod_{k=p}^{n}\Phi(2q_{k-1},2\tau(q_{k-1})-2\eta\varepsilon) \\
&&\qquad \qquad \cdot \|Y_q(w)\|_2 \cdot \|Z_L^{\eta\varepsilon}\|_2 \cdot \||W_r|^\eta\|_2 \cdot (\||W_i|^\eta\|_2)^{n-p} .
\end{eqnarray*}
Also, by using the same arguments as in the proof of (1) we can get 
$$
\begin{cases}\log \Phi(q_{k-1},\tau(q_{k-1})-\eta\varepsilon)= -\widetilde\xi(q_{k-1}) \eta\varepsilon +O(\eta\varepsilon^2)\\
{\scriptstyle \log \Phi(2q_{k-1},2\tau(q_{k-1})-2\eta\varepsilon) \le \log \Phi(2q_{k-1},\tau(2q_{k-1})-2\eta\varepsilon) =-\widetilde\xi(2q_{k-1}) \eta\varepsilon +O(\eta\varepsilon^2).}
\end{cases} 
$$
It follows that
\begin{eqnarray*}
\mathcal{Q}_q(X_{n,p}< -\varepsilon) \le C\cdot e^{n[-\underline{a}\eta\varepsilon+O(\eta^2)]}\cdot (\||W_i|^\eta\|_2)^{\alpha n}.
\end{eqnarray*}
Since $\||W_i|^\eta\|_2\to 1$ when $\eta\to 0$, then we can find $\eta$ small enough and $\alpha$ small enough such that for $n$ large enough:
$$
\mathcal{Q}_q(X_{n,p}< -\varepsilon) \le e^{-\underline{a}\eta\varepsilon n/2},\quad \forall \ (1-\alpha)n\le p\le n.
$$
Consequently, $\sum_{n\ge 1}\mathcal{Q}_q(\exists\  (1-\alpha)n\le m\le n: X_{n,p}< -\varepsilon)<\infty$, and the conclusion follows from the Borel-Cantelli lemma.

\medskip

\noindent
{\it Proof of Proposition~\ref{endpoints2}.} Due to Proposition~\ref{endpoints3}(1), with probability 1, for $\mu_q$ almost every $t\in\mathscr{A}^{\mathbb{N}_+}$, (notice that $1/h$ can be infinite since $h$ can be equal to $0$),
$$
\lim_{n\to\infty} \frac{\log |Q_W(t|_{n})|}{\log Q_L(t|_{n})}=\lim_{n\to\infty} \frac{\log |Q_W(t|_{n})|}{\log\mathrm{Osc}^{(1)}_{F_L}(t|_{n})}=h, \lim_{n\to\infty} \frac{\log Q_L(t|_{n})}{\log\mathrm{Osc}^{(m)}_{F_W}(t|_{n})}=\frac{1}{h}, 
$$
and for $\gamma\in\{-1,1\}$
$$
\lim_{n\to\infty} \frac{\log Z_W^{(1)}((t|_{n})^\gamma)}{\log\mathrm{Osc}^{(1)}_{F_L}(t|_{n})}=\lim_{n\to\infty} \frac{\log Z_L((t|_{n})^\gamma)}{\log\mathrm{Osc}^{(m)}_{F_W}(t|_{n})}=0.
$$
Also, due to the Lemma~\ref{controlI} and the fact that all the moments of $Z_L$ are finite, there exist $\varepsilon>0$ such that for  $\mu_q$-almost every $t\in\mathscr{A}^{\mathbb{N}_+}$, there exists $n_{t,\varepsilon}$ such that for all $n\ge n_{t,\varepsilon}$ we have $ Q_L(t|_{n})\in [b^{-n(\bar{a}+\varepsilon)},b^{-n(\underline{a}-\varepsilon)}] $. In particular, for $n$ large enough we have
$$
\frac{\log (Q_L(t|_{n})/Q_L(t|_{n-N_n(t)}))}{\log Q_L(t|_{n})} \in \Big [\frac{\underline{a}-\varepsilon}{\bar{a}+\varepsilon}\frac{N_n(t)}{n},\frac{\bar{a}+\varepsilon}{\underline{a}-\varepsilon}\frac{N_n(t)}{n}\Big ].
$$
Consequently, since $N_n(t)=o(n)$ for $\mu_q$-almost every $t\in\mathscr{A}^{\mathbb{N}_+}$ (Proposition~\ref{endpoints3}(2)), we have
$$
\lim_{n\to\infty}\frac{\log Q_L(t|_{n-N_n(t)})}{\log Q_L(t|_{n})}=1,
$$
and if $h\neq 0$, we have
$$
1=h\cdot 1\cdot \frac{1}{h}=\lim_{n\to\infty}\frac{\log Q_W(t|_{n-N_n(t)})}{\log Q_L(t|_{n-N_n(t)})}\cdot \lim_{n\to\infty}\frac{\log Q_L(t|_{n-N_n(t)})}{\log Q_L(t|_{n})}\cdot\lim_{n\to\infty}\frac{\log Q_L(t|_{n})}{\log Q_W(t|_{n})}$$
$$
=\lim_{n\to\infty}\frac{\log Q_W(t|_{n-N_n(t)})}{\log Q_W(t|_{n})}.
$$
Moreover, let $i\in\{0,b-1\}$ and $r\in\{0,...,b-1\}$, since $L_i\le 1$,  for any $p\le n-1$ we have $$\liminf_{n\to \infty}\frac{\log (L_{r}(t|_{p-1})\prod_{k=p+1}^{n-1}L_i(t|_{k}))}{\log\mathrm{Osc}^{(m)}_{F_W}(t|_{n})} \ge 0.$$ Then, due to Proposition~\ref{endpoints3} (1) and (3), for $\mu_q$-almost every $t\in\mathscr{A}^{\mathbb{N}_+}$, for $\gamma\in\{-1,1\}$,
$$
\text{either } \liminf_{n\to\infty} \frac{\log \mathrm{Osc}^{(1)}_{F_W}((t|_{n})^\gamma)}{\log\mathrm{Osc}^{(1)}_{F_L}(t|_{n})}=\infty \text{ or } \liminf_{n\to\infty} \frac{\log \mathrm{Osc}^{(1)}_{F_W}((t|_{n})^\gamma)}{\log\mathrm{Osc}^{(1)}_{F_L}(t|_{n})}\ge h;
$$
where the inequality is automatically true in the case $h=0$, and
$$\liminf_{n\to\infty} \frac{\log \mathrm{Osc}^{(1)}_{F_L}((t|_{n})^\gamma)}{\log\mathrm{Osc}^{(m)}_{F_W}(t|_{n})} \ge \frac{1}{h}.
$$
We conclude from the fact that due to~\eqref{omo1}, for $\widetilde{t}=F_L(\pi(t))$ we have
\begin{eqnarray*}
\liminf_{r\to 0} \frac{\log \mathrm{Osc}^{(m)}_{F}(B(\widetilde{t},r))}{\log r} &\ge& \min_{\gamma=-1,0,1} \liminf_{n\to\infty} \frac{\log \mathrm{Osc}^{(1)}_{F_W}((t|_{n})^\gamma)}{\log\mathrm{Osc}^{(1)}_{F_L}(t|_{n})}; \\
\liminf_{r\to 0} \frac{\log r}{\log \mathrm{Osc}^{(m)}_{F}(B(\widetilde{t},r))}  &\ge & \min_{\gamma=-1,0,1} \liminf_{n\to\infty} \frac{\log \mathrm{Osc}^{(1)}_{F_L}((t|_{n})^\gamma)}{\log\mathrm{Osc}^{(m)}_{F_W}(t|_{n})},
\end{eqnarray*}
where in the last inequality we have used Lemma~\ref{controlI}. Consequently,
$$\lim_{r\to 0} \frac{\log \mathrm{Osc}^{(m)}_{F}(B(\widetilde{t},r))}{\log r} =h, \text{ for }\mu_q\circ\pi^{-1}\circ F_L^{-1} \text{-almost every } \widetilde{t}.$$

\medskip

\noindent
{\it Almost sulely $\dim_H (\nu_q) \geq \tau^*(h)$.} We only need to deal with the case where $\tau^*(h)>0$. We are going to prove that, with probability 1, for $\mu_{q}$-almost every $t\in\mathscr{A}^{\mathbb{N}_+}$, 
$$\liminf_{n\to\infty} \frac{\log \mu_{q}([t|_{n}])}{\log O_L(t|_{n})} \ge \tau^*(h).$$
Then, due to the last claim of Lemma~\ref{controlI}, the mass distribution principle (see \cite{Pe}, Lemma 4.3.2 or Section 4.1 in \cite{Falc}) yields the conclusion.

We set $d=\tau^*(h)$. Fix  $\varepsilon>0$, and for $n\geq 1$ define
$$E_{n,\varepsilon}=\{t\in\supp(\mu_q): O_L(t|_{n})^{-d+\varepsilon}\cdot \mu_{q}([t|_{n}])\geq 1\}.$$
For $n\ge 1$ let $\eta_n>0$ and set $\displaystyle
S_{n,\varepsilon}=\sum_{w\in \mathscr{A}^n} \mu_{q}([w])\Big(O_L(t|_{n})^{-d+\varepsilon}\cdot \mu_{q}([w])\Big)^{\eta_n}$. For $n\geq 1$ we have $\displaystyle 
\mu_{q}(E_{n,\varepsilon})\leq S_{n,\varepsilon}= \sum_{w\in \mathscr{A}^n}Q_L(w)^{-(d-\varepsilon)\eta_n}Q_q(w)^{1+\eta_n} Y_q(w)^{1+\eta_n}Z_L(w)^{\eta_n}$ and 
$\displaystyle
\mathbb{E}(S_{n,\varepsilon})=\prod_{k=0}^{n-1}\Phi(q_k+q_k\eta_n,\tau(q_k)+(\tau(q_k)+d-\varepsilon)\eta_n)\cdot \mathbb{E}(Y_q(w)^{1+\eta_n}Z_L(w)^{\eta_n}) .$
If we show that $\sum_{n\ge 1} \mathbb{E}(S_{n,\varepsilon})<\infty$, then the series $\sum_{n\ge 1} \mu_q(E_{n,\varepsilon})$ converges almost surely and the conclusion follows from the Borel Cantelli lemma.

By an argument similar to those used in the proof of Proposition~\ref{endpoints}, we have
\begin{eqnarray*}
&&\log \Phi(q_k+q_k\eta_n,\tau(q_k)+(\tau(q_k)+d-\varepsilon)\eta_n) \\
&=&-\xi(q_k)q_k\eta_n+\widetilde{\xi}(q_k)(\tau(q_k)+d-\varepsilon)\eta_n+O(\eta_n^2)\\
&=&-\widetilde{\xi}(q_k)(\tau^*(\tau'(q_k))-d+\varepsilon)\eta_n+O(\eta_n^2)\le -\underline{a}\varepsilon\eta_n+O(\eta_n^2)
\end{eqnarray*}
Thus
$$\prod_{k=0}^{n-1}\Phi(q_k+q_k\eta_n,\tau(q_k)+(\tau(q_k)+d-\varepsilon)\eta_n) \le e^{-\underline{a}\varepsilon\eta_nn+O(n\eta_n^2)}.$$

Now since
$$
\mathbb{E}(Y_q(w)^{1+\eta_n}Z_L(w)^{\eta_n}) \le \|Y_q(w)\|_2^{1+\eta_n}\cdot \|Z_L(w)^{\frac{\eta_n}{1-\eta_n}}\|_2^{1-\eta_n},
$$
by taking $\eta_n=\frac{1}{\sqrt{n_0+n}}$ for $n_0$ large enough we will get $\mathbb{E}(S_{n,\varepsilon})=O(b^{-\frac{\underline{a}\varepsilon n}{2\sqrt{\log (n)}}})$.

\appendix
\section{Appendix}
\label{appendix}

\begin{proposition}\label{Appendix}
Let $M$ be a non-negative bounded and non-decreasing function defined over the subsets of $\mathbb{R}^d$. Let $\supp(M)=\{t: \ \forall r>0, \ M(B(t,r))>0\}$ be the closed support of $M$. Suppose that $\supp(M)$ is a non-empty compact set and define the $L^q$-spectrum associated with $M$ as the mapping namely 
$$
\tau_M(q)=\liminf_{r\to 0}\frac{\log \sup \left\{\sum_i M(B_i)^q\right \}}{\log (r)},
$$
where the supremum is taken over all the  families of  disjoint closed balls $B_i$ of radius $r$ with centers in $\supp(M)$. We have $\overline{\dim}_B(\supp(M))=-\tau_M (0)$, and for all $h\ge 0$, 
$$
\dim\, E_M(h):=\Big \{t\in\supp(M):\liminf_{r\to 0^+}\frac{\log (M(B(t,r))}{\log (r)}=h\Big \}\le \tau_M^*(h),
$$
a negative dimension meaning that $E_M(h)$ is empty. 
\end{proposition}
\begin{proof}
The equality $\overline{\dim}_B(\supp(M))=-\tau_M (0)$ is just the definition of the upper box dimension. 

\medskip

Let $h\ge 0$.  Fix $\varepsilon>0$.  For every $t\in E_M(h)$, let $(r_{t,k})_{k\ge 0}$ be a decreasing sequence tending to 0 such that $r_{t,k}^{h+\varepsilon}\le M(B(t, r_{t,k}))\le r_{t,k}^{h-\varepsilon}$.  

Fix $\delta>0$, and for each $t\in E_M(h)$ let $k_t$ be such that $r_{t,k_t}\le \delta$. Now, for every $n\ge 0$, let $A_n=\{t\in E_M(h): 2^{-(n+1)}<r_{t,k_t}\le 2^{-n}\}$. By the Besicovich covering theorem (see Theorem 2.7 in \cite{Mat}) there exists an integer $N$ such that for every $n\ge 0$ we can find $N$ disjoint subsets $A_{n,1},\dots A_{n,N}$ of $A_n$ such that each set $A_{n,j}$ is at most countable, the balls of the form $B(t,  r_{t,k_t})$, $t\in A_{n,j}$, are pairwise disjoint, and $\bigcup_{n\ge 0}\bigcup_{j=1}^N\bigcup_{t\in A_{n,j}}B(t,  r_{t,k_t})$ is a $\delta$-covering of $E_M(h)$. 

Suppose that $h\in [0,\tau_f'(0^+)]$. We have $\tau_M^*(h)=\inf_{q\in\mathbb{R}_+} hq-\tau_M(q)$. Fix $q\ge 0$ such that $\tau_M(q)>-\infty$ and then define $D_\varepsilon=(h+\varepsilon) q-\tau_M(q)+\varepsilon$. We have 
\begin{multline*}
\mathcal{H}^{D_\varepsilon}_\delta(E_M(h))\le \sum_{n\ge 0}\sum_{j=1}^N\sum_{t\in A_{n,j}} (2r_{t,k_t})^{D_\varepsilon}
\le 2^{D_\varepsilon}\sum_{n\ge 0}\sum_{j=1}^N\sum_{t\in A_{n,j}}  r_{t,k_t}^{(h+\varepsilon) q-\tau_M(q)+\varepsilon}\\
\le  2^{D_\varepsilon}\sum_{n\ge 0}\sum_{j=1}^N\sum_{t\in A_{n,j}} M(B(t, r_{t,k}))^q r_{t,k_t}^{-\tau_M(q)+\varepsilon}\\
\le  2^{D_\varepsilon}2^{|\tau_M(q)|}\sum_{n\ge 0}\sum_{j=1}^N\sum_{t\in A_{n,j}} M(B(t,2^{-n}))^q 2^{n(\tau_M(q)-\varepsilon)}.
\end{multline*}
For each $1\le j\le N$, the family $\{B(t,2^{-n})\}_{t\in A_{n,j}}$ can be divided into two disjoint $2^{-n}$-packing of $\supp(M)$. Consequently, by definition of $\tau_M(q)$, for $n$ large enough, 
$$
\sum_{t\in A_{n,j}} M(B(t,2^{-n}))^q\le 2\cdot 2^{-n (\tau_M(q)-\varepsilon/2)}
$$
and $\mathcal{H}^{D_\varepsilon}_\delta(E_M(h))=O\big (\sum_{n\ge 0} 2^{-n\varepsilon/2}\big )<\infty$. This yields $\dim\, E_M(h)\le D_\varepsilon$ for all $\varepsilon>0$, hence $\dim\, E_M(h)\le hq-\tau_M(q)$.

\medskip

 Now suppose that $h>\tau_f'(0^+)$. We have $\tau_M^*(h)=\inf_{q\in\mathbb{R}_-} hq-\tau_M(q)$. Fix $q\le  0$ such that $\tau_M(q)>-\infty$ and then $D_\varepsilon=(h-\varepsilon) q-\tau_M(q)+\varepsilon$. This time we have 
$$
\mathcal{H}^{D_\varepsilon}_\delta(E_M(h))\le 2^{D_\varepsilon}2^{|\tau_M(q)|}\sum_{n\ge 0}\sum_{j=1}^N\sum_{t\in A_{n,j}} M(B(t,2^{-(n+1)}))^q 2^{n(\tau_M(q)-\varepsilon)},
$$
and for each $1\le j\le N$, the family $\{B(t,2^{-(n+1)})\}_{t\in A_{n,j}}$ is a $2^{-(n+1)}$-packing of $\supp(M)$. We conclude as in the previous case. 
\end{proof}

\begin{proposition}\label{VBE}
Let $\big (V^{(n)}=(V^{(n)}_0,\dots,V^{(n)}_{b-1})\big )_{n\ge 1}$, be a sequence of random vectors taking values in $\mathbb{C}^b$, and such that $\mathbb{E}\big (\sum_{i=0}^{b-1}V^{(n)}_i\big )=1$. Let $\{V(w)\}_{w\in\mathscr{A}^*}$ be a sequence of independent vectors such that $V(w)$ is distributed as $V^{(|w|)}$ for each $w\in\mathscr{A}^*$. Define $Z_0=1$ and for $n\ge 1$
$$
Z_n=\sum_{w\in\mathscr{A}^n}\prod_{k=1}^{n}V_{w_k}(w|_{k-1}).
$$
Let $p\in (1,2]$. There exists a constant $C_p\le 2^p$ depending on $p$ only such that for all $n\ge 1$
$$
\mathbb{E}(|Z_n-Z_{n-1}|^p)\le C_p\prod_{k=1}^n\mathbb{E}\Big (\sum_{i=0}^{b-1}|V_i^{(k)}|^p\Big ).
$$
\end{proposition}
See the proof of Theorem 1 in \cite{B3}. 

\begin{proposition}\label{xmoments} We work under the assumptions of Theorem~A or B. Let $m\ge 1$ and $U\in\{W,L\}$.

(1) If $q>1$ and $\varphi_U(q)>0$ then $\mathbb{E}((Z_U^{(m)})^q)<\infty$. Moreover, if $W$ satisfies the assumptions of Theorem~B(2) then  ${\rm ess\, sup}\, {\rm Osc}^{(m)}_{F_W}([0,1])<\infty$.

(2) Define $\psi_U^{(m)}(t)=\mathbb{E}(e^{-tZ_U^{(m)}})$ for $t\geq 0$. Let $A_U=\max_{0\le i\le b-1} |U_i|$.

If $q>0$ and $\mathbb{E}(A_U^{-q})<\infty$ then $\psi_U^{(m)}(t) = O(t^{-p})$ for all $p\in(0,q)$. Consequently, $\mathbb{E}((Z_U^{(m)})^{-p})<\infty$ for all $p\in (0,q)$. 
\end{proposition}

\noindent
{\it{Proof of Proposition $\ref{xmoments}$}}  (1) Since ${\rm Osc}^{(m)}_{F_U}([0,1])\le 2^{m-1}{\rm Osc}^{(1)}(F_U,[0,1])\le 2^m\|F_U\|_\infty$, this is a direct consequence of Theorems A and B (that ${\rm ess\, sup}\, {\rm Osc}^{(m)}_{F_W}([0,1])<\infty$ when $W$ satisfies the assumptions of Theorem~B(2) is not stated in \cite{BJMpartII} but established in the proof of this theorem).

\noindent
(2) Since ${\rm Osc}^{(m)}_{F_U}([0,1])\ge{\rm Osc}^{(m)}_{F_U}(I_i)$ for all $0\le i\le b-1$, by using (\ref{self-sim2}) we get
\begin{equation}\label{xineq}
Z_U^{(m)} \geq  b^{-1}\sum_{i=0}^{b-1} |U(i)|\cdot Z_U^{(m)}(i),
\end{equation}
where the $Z_U^{(m)}(i)$ are independent copies of $Z$ and they are independent of $W$. 

Moreover, thanks to Proposition \ref{nonpol} applied to $F_U$, we know that $Z_U^{(m)}>0$ almost surely for all $m\in\mathbb{N}_+$. Also, with probability 1, we can define $i_0=\max\{0\le i\le b-1:|U_i|=\max_{0\le k\le b-1}|U_k|$ and  $i_1=\inf\{0\le i \le b-1: i\neq i_0,\ U_i\neq 0\}$, $A_0=|U_{i_0}|$ and $A_1=|U_{i_1}|$. 

Suppose that $\mathbb{E}(A_0^{-q})<\infty$. This clearly holds if $\varphi_U(-q)>-\infty$  or if there exists $a>0$ such that $ \max_{0\le k\le b-1}|U_k|\ge a$ almost surely (for instance  $a=1/b$ is convenient when $U$ is conservative). 

 Set $\psi=\psi_U$. By definition of $\psi$, we deduce from (\ref{xineq}) and the fact that $Z_U^{(m)}$ is almost surely positive that $\psi(t)\le \mathbb{E}\big (\psi(A_0 t)\psi(A_1 t)\big )$ and $\lim_{t\to\infty} \psi(t)=0$. Suppose that we have shown that $\psi(t)=O(t^{-p})$ at $+\infty$, for all $p\in(0,q)$. Then, for $x>0$ we have $\mathbb{P}(Z^{(m)}\le x)\le e^{tx}\psi(t)$ and choosing $t=p/x$ yields $\mathbb{P}(Z\le x)=O(x^{p})$ at $0^+$. Hence $\mathbb{E}(Z^{-p})<\infty$ if $h\in (0,q)$.  

Now we essentially use the elegant approach of \cite{Liu:1} for the finiteness of the moments of negative orders of $F_U(1)$, when the components of $W$ are non-negative (see also the references in \cite{Liu:1}  for this question). Let $r>1$ and $\phi=\psi^r$. Due to the bounded convergence theorem we have $\lim_{t\to\infty} \mathbb{E}(\psi(A_1t)^{r/(r-1)})=0$, so the H\"older inequality yields $\phi(t)=o(\mathbb{E}(\phi(t A_0))$ at $\infty$. Let $\gamma\in (0,1)$ small enough to have $\gamma\mathbb{E}(A_0^{-p})<1$, and let $t_0>0$ such that 
\begin{equation}\label{momneg}
\phi(t)\le \gamma\mathbb{E}(\phi(t A_0)),\  t\ge t_0. 
\end{equation}
Let $(\widetilde A_i)_{i\ge 1}$ be a sequence of independent copies of $A_0$ Since $\phi\le 1$, for $t\ge t_0$ we can prove by induction using (\ref{momneg}) the following inequalities valid for all $n\ge 2$:
\begin{eqnarray*}
\phi(t)&\le & \gamma\mathbb{P}(A_0t<t_0)+\gamma\mathbb{E}\big (\mathbf{1}_{\{A_0t\ge t_0\}}\phi(A_0 t)\big)\\
&\le &\gamma\mathbb{E}(A_0^{-p}) (t_0/t)^p+\gamma^2 \mathbb{E}\big ( \mathbf{1}_{\{A_0t\ge t_0\}}\phi(A_0 \widetilde A_1t)\big)\\
&\le &\gamma\mathbb{E}(A_0^{-p}) (t_0/t)^p +\gamma^2 \mathbb{E}\big (\phi(A_0 \widetilde A_1t)\big)\\
&\le & \gamma\mathbb{E}(A_0^{-p}) (t_0/t)^p+\gamma^2(\mathbb{E}(A_0^{-p}))^2 (t_0/t)^p+\gamma^2 \mathbb{E}( \mathbf{1}_{\{A_0\widetilde A_1t\ge t_0\}}\phi(A_0 \widetilde A_1t)\big)\\
&\le & (t_0/t)^p\sum_{k=1}^n(\gamma\mathbb{E}(A_0^{-p}))^k+\gamma^n\mathbb{E}( \mathbf{1}_{\{A_0\widetilde A_1\cdots \widetilde A_{n-1}t\ge t_0\}}\phi(A_0 \widetilde A_1\cdots \widetilde A_{n-1}t)\big).
\end{eqnarray*}
Since $\psi\le 1$, and both $\gamma$ and $\gamma\mathbb{E}(A_0^{-p})$ belong to $(0,1)$, letting $n$ tend to $\infty$ yields $\phi(t)=\psi(t)^r=O(t^{-p})$. Since $r$ and $p$ are arbitrary respectively in $(1,\infty)$ and $(0,q)$, we have the desired result.

%
%

\end{document}